\documentclass[12pt]{article}
\pdfoutput=1
\usepackage{amsmath,amsthm,amssymb,amsfonts,amscd,graphicx}
\usepackage{graphics}
\usepackage[margin=1in]{geometry}
\usepackage[pagebackref]{hyperref} 

\newtheorem{proposition}{Proposition}
\theoremstyle{definition}
\newtheorem{remark}{Remark}
\newtheorem{example}{Example}
\newtheorem{algorithm}{Algorithm}

\renewcommand{\P}{\mathcal{P}}
\renewcommand{\H}{\mathcal{H}}

\newcommand{\R}{\mathbb{R}}

\newcommand{\argmin}{\mathrm{argmin}}
\newcommand{\argmax}{\mathrm{argmax}}
\newcommand{\sign}{\mathrm{sign}}
\newcommand{\shrink}{\mathrm{shrink}}

\renewcommand{\d}{\partial}

\begin{document}

\title{A Robust Time Series Model with Outliers and Missing Entries}
\author{Triet M. Le\thanks{NGA Research, The National Geospatial-Intelligence Agency, Springfield, VA. Email: Triet.M.Le@NGA.mil.}}
\maketitle

\abstract{
This paper studies the problem of robustly learning the correlation function for a univariate time series with the presence of noise, outliers and  missing entries. The outliers or anomalies considered here are sparse and rare events that deviate from normality which is depicted by a correlation function and an uncertainty condition. This general formulation is applied to univariate time series of event counts (or non-negative time series) where the correlation is a log-linear function with the uncertainty condition following the Poisson distribution. Approximations to the sparsity constraint, such as $\ell^r, 0< r\le 1$, are used to obtain robustness in the presence of outliers. The $\ell^r$ constraint is also applied to the correlation function to reduce the number of active coefficients. This task also helps bypassing the model selection procedure. Simulated results are presented to validate the model.
}

\section{Introduction}

Forecasting (anticipating) future events or activities is an important problem in data science. A common task for a forecaster is to predict normal future events using past and current observations and to alert when the observed number of events significantly deviates from the predicted value. If events and activities are random then there is no hope in making any meaningful future prediction. However, if events are correlated in the sense that events at time $t$ depend on events prior to $t$ and the correlation function that describes this dependency persists throughout the observed data, then this correlation function can be used to predict future events based on past and current observations. In many real datasets, the observed data is incomplete and is often contaminated with outliers and random noise.  An important task is then to robustly learn the correlation function that describes the underlying normal activities and patterns from the observed data. We start with the following setup.

Let $\{t_0,\cdots, t_N\}$ be a uniform discretization of some time interval of interest which we assume to be $[0,T)$. Let $y_i$ be the number of observed events and $u_i$ be the expected number of events that occur in the time interval $[t_{i-1},t_i)$. Observed events are determined by the conditional probability
\begin{equation}\label{c1}
 P(y_i | u_i,\theta_i)
\end{equation}
where $\theta_i$ is some auxiliary variable, representing for instance the variance, that $y_i$ may depend on. In a univariate case, the goal is to learn how $u_i$ depends on prior $u_j$ and $y_j$, for $ j<i$. In other words, one is interested in finding the function $f$ such that
\begin{equation}\label{c2}
u_i = f(\{u_j\}_{j<i},\{y_j\}_{j<i}).
\end{equation}

In this paper, we consider the case where only a partial series $\widetilde{y}_D:=\{\widetilde{y}_i\}_{i\in D}$ for some $D\subset \{1,\cdots, N\}$ is observed and it may contain outliers and anomaly. To tackle this problem, the following minimization problem is proposed
\begin{equation}\label{main_min_prob_eq}
\begin{split}
&\min_{a_0,a,b,y} J(a_0,a,b,y),\mbox{ where }\\
J(a_0,a,b,y) &= \sum_{i=1}^N \left[u_i - y_i\log(u_i) + \log(\Gamma(y_i+1)) \right]+ \lambda\sum_{i\in D}|y_i - \tilde{y}_i|^{r}\\
&+ \mu \left(\sum_{k=1}^p |a_k|^{s} + \sum_{k=1}^q |b_k|^{s}\right),
\end{split}
\end{equation}
for some $0<r,s\le 1$, and $\mu,\lambda>0$, to jointly learn $f, \theta:=\{\theta_i\}_{i=1}^N$ and the complete series $y:=\{y_i\}_{i=1}^N$ via imputation. 
\begin{itemize}
\item Here, we consider parametrically the correlation function $f$ defined as
$$
u_i = f\left(\{y_j\}_{j\le i-1},\{u_j\}_{j\le i-1}\right) = \max(\exp(\log(u_i+1)) - 1,0), 
$$
where $\log(u_i+1)$ is given by
\begin{equation}\label{u_eq}
\begin{split}
\log(u_i + 1) = a_0 + \sum_{k=1}^p a_k \log(y_{i-k}+1),
\end{split}
\end{equation}
 for some $p,q\ge 0$. 
 \item $\sum_{i\in D}|y_i - \tilde{y}_i|^{r}$ measures the sparsity of the sequence $\{y_i-\widetilde{y}_i\}_{i\in D}$ representing anomaly. Similarly, $\sum_{k=1}^p |a_k|^{s}$ and $\sum_{k=1}^q |b_k|^{s}$ impose sparsity on $a=\{a_k\}_{k=1}^p$ and $b=\{b_k\}_{k=1}^q$ and overcome the model selection issue (e.g. AIC, BIC, etc.). Both of these constraints are important for the recovery of $f$.
\end{itemize}
 With the presence of missing entries and outliers, Figure \ref{fig0} shows the ability of the model \eqref{main_min_prob_eq} to recover $f$ having
 \begin{eqnarray*}
a_0 &=& 1,\\
a=(a_1,a_2,a_3,a_4,a_5) 
&=& (0.25,-0.5,0, 0,-0.5,0.5).
\end{eqnarray*}

\begin{figure}
\begin{center}
\includegraphics[scale=0.45]{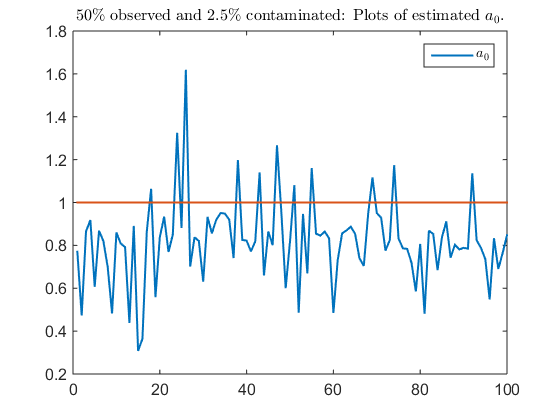}
\includegraphics[scale=0.45]{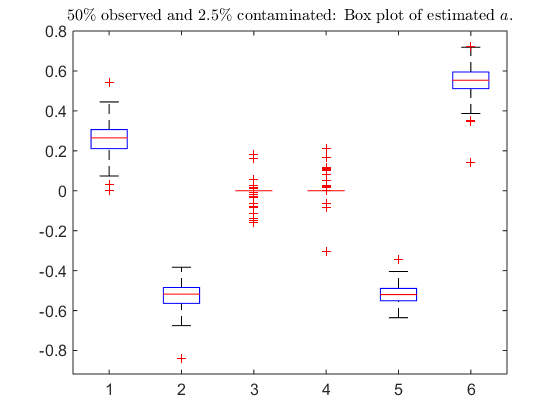}
\caption{Plots of estimated parameters $a_0$ and $a$ for 100 simulations when $50\%$ of the data is observed and $2.5\%$ of the observed entries are contaminated.}
\label{fig0}
\end{center}
\end{figure}


In the followings, we provide motivations for the proposed model. Many well-known time series models can be described by \eqref{c1} and \eqref{c2}. For instance, in an $ARMA(p,q)$ model \cite{box1970time,george1994time}, each $y_i$ is defined as
\begin{equation}\label{arma_eq1}
y_i = \phi_0 + \sum_{k=1}^p \phi_k y_{i-k} + \sum_{k=1}^q \psi_k e_{i-k} + e_i,
\end{equation}
where $e_i$ follows $N(0,\sigma^2)$, a normal Gaussian distribution with mean $0$ and variance $\sigma^2$. Let 
\begin{equation}\label{arma_u_eq}
u_i := \phi_0 + \sum_{k=1}^p \phi_k y_{i-k} + \sum_{k=1}^q \psi_k e_{i-k},
\end{equation}
then equation \eqref{arma_eq1} implies $y_i = u_i + e_i$. This implies
\begin{equation}\label{arma_c1_eq}
\begin{split}
P(y_i|u_i,\theta_i = \sigma) &= P(y_i-u_i = e_i)\\
 &=\frac{1}{\sigma\sqrt{2\pi}}e^{-\frac{|y_i-u_i|^2}{2\sigma^2}}.
\end{split}
\end{equation}
Substituting $y_{i-k} - u_{i-k}$ for $e_{i-k}$ in \eqref{arma_u_eq}, one obtains
\begin{equation}\label{arma_c2_eq}
u_i = f(\{u_j\}_{j<i},\{y_j\}_{j<i}) =  \phi_0 + \sum_{k=1}^{p}a_k y_{i-k} + \sum_{k=1}^q b_k u_{i-k},
\end{equation}
for some real-valued $a_k$ and $b_k$ and some new positive integers $p$ and $q$. Here we assume zero boundary conditions, that is $y_{i-k}$ and $u_{i-k}$ are identically zero whenever $i-k \le 0$. Clearly, other types of boundary conditions such as reflection or Neumann can be used. Thus \eqref{arma_eq1} can be transformed into \eqref{arma_c1_eq} and \eqref{arma_c2_eq}, which are \eqref{c1} and \eqref{c2} respectively.

Another example is the Poisson linear autoregressive model (see \cite{harvey1989time,li1994time}, among others). Assuming $y_i$ is a non-negative integer, then \eqref{c1} is given by 
\begin{equation}\label{par_c1}
P(y_i|u_i) = \frac{u_i^{y_i} e^{-u_i}}{y_i!},
\end{equation}
and \eqref{c2} is given by 
$$
u_i = f(\{u_j\}_{j<i},\{y_j\}_{j<i}) =  a_0 + \sum_{k=1}^p a_k y_{i-k} + \sum_{k=1}^q b_k u_{i-k},
$$
where $a_0,a_k,$ and $b_k$ are non-negative. This shows that the model can only detect zero or positive correlations. To overcome this drawback, the log-linear model is often used \cite{zeger1988markov,harvey1989time,li1994time,fokianos2012count} where $u_i$ is defined such that
$$
\log(u_i+1) = a_0 + \sum_{k=1}^p a_k \log(y_{i-k}+1) + \sum_{k=1}^q b_k \log(u_{i-k}+1).
$$
Here $a_0$, $a_k$ and $b_k$ are real-valued and therefore can represent negative correlations. 

In the case of complete observations where we are given the series $y=\{y_i\}_{i=1}^N$ and some prior knowledge on the conditional probability condition \eqref{c1}, the task is then to learn the optimal $(f^*,\theta^*)$ that maximizes the likelihood function,
\begin{equation}\label{comp_max_eq}
(f^*,\theta^*) = \argmax_{f,\theta} \left\{L(f,\theta)  = P(y, f,\theta)\right\}.
\end{equation}

It can be shown (see for instance \cite{lutkepohl2007new}) that the joint probability in \eqref{comp_max_eq} is given by
\begin{equation}
P(y, f,\theta) = \prod_{i=1}^N P(y_i|u_i,\theta_i) P(f)P(\theta),
\end{equation}
assuming $f$ and $\theta$ are independent. The maximization problem in \eqref{comp_max_eq} is equivalent to
\begin{equation}\label{com_min_eq}
(f^*,\theta^*) = \argmin_{f,\theta} \left\{ -\sum_{i=1}^N \log(P(y_i|u_i,\theta_i)) - \log(P(f)P(\theta)) \right\}.
\end{equation}

With some knowledge about $\theta_{N+1}$, the conditional probability $P(y_{N+1}|u_{N+1},\theta_{N+1})$ provides the distribution of $y_{N+1}$. For a single value prediction $\hat{y}_{N+1}$, we can solve
\begin{equation}
\hat{y}_{N+1} := \argmax_{y_{N+1}} P(y_{N+1}|u_{N+1},\theta_{N+1}).
\end{equation}
In case where the conditional probability follows a Poisson distribution, then $y_{N+1}$ is completely determined by $u_{N+1}$, and it doesn't depend on the auxiliary variable $\theta_{N+1}$.

In this paper, we consider the problem of learning parametrically the underlying  correlation function $f$ and $\theta$ given a partially observed series $\widetilde{y}_D=\{\widetilde{y}_i\}_{i\in D}$ for some $D\subset \{1,\cdots, N\}$ which may be contaminated by outliers and anomalies. Given some prior knowledge about the uncertainty condition \eqref{c1}, the problem consists of: 1) extending $\widetilde{y}_D$ to the whole series $y$ (including $y_i, i\in D^c$) via imputation such that $y_i\approx \widetilde{y}_i$ for all $i\in D$ and 2) using the complete series $y$ to learn $f$. These two steps are done iteratively as they are inter-dependent. The interpretation of $y_i\approx\widetilde{y}_i$ is as follows:
\begin{enumerate}
\item Suppose $\widetilde{y}_i$ is normal, i.e. it can be described by $f,\theta$ and the uncertainty condition \eqref{c1}, then we would like to enforce $y_i = \widetilde{y}_i$.
\item On the other hand if $\widetilde{y}_i$ is anomalous, then we allow $y_i\neq \widetilde{y}_i$ and let the model decides a normal value for $y_i$. 
\end{enumerate}
Moreover, outliers and anomalies are interpreted as rare events that are supported sparsely on $D$. Based on this interpretation, we would like the difference series $\{y_i-\widetilde{y}_i\}_{i\in D}$ to be sparse. Thus this can be seen as solving the optimization problem:
\begin{equation}\label{log_gen_eq}
\min_{y,f,\theta} \left\{ -\sum_{i=1}^N \log(P(y_i|u_i,\theta_i)) - \log(P(f)) - \log(P(\theta))\right\},
\end{equation}
with the constraint that the partial series $\{y_i-\widetilde{y}_i\}_{i\in D}$ is sparse. Here, $P(f)$ and $P(\theta)$ are priors on $f$ and $\theta$ respectively. Sparsity is an essential ingredient in the theory of compressed sensing \cite{candes2006robust, candes2006stable,donoho2006compressed}. Approximating the sparsity constraint has been a subject of great importance as the exact sparsity problem is NP-hard. Here we consider sparsity approximations as proposed in \cite{candes2006robust, candes2006stable, donoho2006compressed,chartrand2007exact} by using $\ell^r$ for $0<r\le 1$. Incorporating these sparsity approximations into the minimizing energy \eqref{log_gen_eq}, we propose the following unconstraint variational problem
\begin{equation}\label{min_prob_eq}
\begin{split}
\min_{f,y,\theta}\Big\{
J(y,f,\theta) &=  -\sum_{i=1}^N \log(P(y_i|u_i,\theta_i)) - \log(P(f)) - \log(P(\theta))\\
& + \lambda \sum_{i\in D}|y_i - \tilde{y}_i|^r
\Big\}.
\end{split}
\end{equation}

\begin{remark}
Predicting $y_{N+1}$ with the mean $u_{N+1}$ is optimal. However, for $k\notin D$ and $k<N$, imputing $y_k$ with the mean $u_k$ is not always optimal. Indeed, suppose $u_k$ only depends on $p$ previous $y_i$'s, i.e. $u_k = f\left(\{y_i\}_{i=k-p}^{k-1}\right)$ for some $p\ge 1$. Suppose also that both $f$ and $\theta$ are known with no outliers, that is $y_i = \tilde{y}_i$ for all $i\in D$ and $y_k$ is the only non-observed entry. The task is to compute the optimal $y_{k}^*$ by using \eqref{min_prob_eq}. This amounts to solving
$$
y_{k}^* = \argmin_{y_{k}} \left[ -\sum_{i=1}^N \log(P(y_i|u_i,\theta_i))\right].
$$
If $k = N$, then it is clear that only the last term in the above sum contains $y_N$. This implies
$$
y_N^* = \argmin_{y_N} \left[-\log(P(y_N|u_N,\theta_N)) \right]= \argmax_{y_N} P(y_N|u_N,\theta_N). 
$$
On the other hand, if $1\le k < N$, then
\begin{eqnarray*}
y_{k}^* &=& \argmin_{y_k} \left[- \sum_{i=k}^{\min(k+p,N)}  \log(P(y_i|u_i,\theta_i))\right] \\
&=& \argmax_{y_k}\left[ \prod_{i=k}^{\min(k+p,N)}P(y_i|u_i,\theta_i)\right].
\end{eqnarray*}
The latter case shows that $\argmax_{y_k} P(y_k|u_k,\theta_k)$ is not always an optimal value for $y_k^*$. 
\end{remark}

The paper is organized as follows. Section \ref{prior_sec} recalls some related prior work that are most relevant. Section \ref{plar_sec} describes the proposed Poisson log-linear model having the uncertainty condition \eqref{c1} following the Poisson distribution and the correlation function \eqref{c2} following a log-linear function. One subproblem for solving \eqref{min_prob_eq} is to compute
$$
\inf_{t\in \R}\left\{ E_r(t)=\mu |t|^r + \frac{1}{2}|t-t'|^2\right\},
$$
for some fixed $t'\in \R$, $\mu>0$ and $0<r\le 1$. In \cite{nie2012robust}, Nie-etal proposed a method for solving this problem via solving a zero of a strictly convex function. For completeness, we go over in section \ref{prox_p_sec} a similar method for computing the proximal operator for $E_r(t)$. Section \ref{numeric_sec} goes over an algorithm to compute a minimizer for \eqref{min_prob_eq}. Section \ref{num_res} shows numerical results on simulated data to validate the proposed model. In Appendix \ref{appen_A} we show that the above minimization problem \eqref{min_prob_eq} is related to 
$$
\max_{y,f,\theta} P(\tilde{y}_D,y,f,\theta),
$$
for some fixed $\lambda>0$ and $r\in (0,1]$.


\section{Prior Work}\label{prior_sec}

In \cite{huber1973robust}, Huber considered the classical least square problem of learning $p$ parameters $a_1,\cdots, a_p$ from $N$ observations $y_1,\cdots, y_N$ obeying the relation
\begin{equation}\label{ols_eq}
y_i = \sum_{j=1}^p x_{ij} a_j + e_i.
\end{equation}
Here $x_{ij}$ are known coefficients and $e_i\approx N(0,\sigma^2)$ are iid random Gaussian noise. For an autoregressive model, $x_{ij} = y_{i-j}$. Estimating the parameter $a=(a_j)_{j=1}^p\in \R^p$ amounts to minimizing the sum of squares
\begin{equation}\label{min_ols}
\min_{a\in \R^p}\left\{\sum_{i=1}^N \left|y_i - \sum_{j=1}^p x_{ij}a_j\right|^2 = \sum_{i=1}^N |e_i|^2\right\}.
\end{equation}

Let $X = (x_{ij})_{N\times p}$, $a = (a_j)_{p\times 1}$ and $u=Xa$, then the relative condition number measuring the sensitivity of $u$ with respect to perturbation of $a$ is \cite{trefethen1997numerical}
$$
\kappa_{a\rightarrow u} = \|X\|\frac{\|a\|}{\|Xa\|} \le \|X\|\|X^+\|
$$
where $\|\cdot\|$ is an arbitrary matrix norm and $X^{+}$ is the pseudo inverse of $X$ if exists.  Let $\sigma_1$ and $\sigma_p$ be the largest and smallest singular values of $X$ respectively, then by using $\|\cdot\| = \|\cdot\|_2$, one has $\kappa_{a\rightarrow u}  = \frac{\sigma_1}{\sigma_p}$. If $\kappa_{a\rightarrow u} $ is large, a small deviation in $a$ can create large deviation in the solution $u$ and hence it is important to obtain an accurate estimate for $a$. As noted in \cite{huber1973robust}, outliers affect the accuracy of the estimates.  Following Lecture 18 in \cite{trefethen1997numerical}, the condition number for measuring the sensitivity of $a$ with respect to perturbation of $y$ is $\kappa_{y\rightarrow a} = \kappa(X)/(\eta \cos(\theta))$, where $\eta = \|X\|\|a\|/\|Xa\|$ and $\theta = \cos^{-1}(\|u\|/\|y\|) = \cos^{-1}(\|u\|/\|u+e\|) $. The more noise and outliers in the system, the closer $\theta$ is to $\pi/2$. This leads to a large value of $\kappa_{y\rightarrow a}$.  Hence it is important to have good estimations of both $a$ and $y$ in the presence of missing data and outliers in the observation.

In \cite{huber1964robust,huber1973robust}, Huber proposed a robust alternative to \eqref{min_ols} by considering
\begin{equation}\label{min_robust}
\min_{a}\left\{\sum_{i=1}^N \rho\left(y_i - \sum_{j=1}^p x_{ij}a_j\right) = \sum_{i=1}^N \rho\left(e_i\right)\right\},
\end{equation}
where $\rho(x)$ is chosen so that it is less sensitive to large $|x|$. In particular, the proposed $\rho$ has the form
\begin{equation}\label{huber_rho}
\rho(x) = \left\{
\begin{matrix}
\frac{1}{2}|x|^2 & \mbox{ if } & |x|< c\\
c|x| - \frac{1}{2}c^2 & \mbox{ if } & |x|\ge c
\end{matrix}
\right.
\end{equation}
where $c$ is some chosen constant which is data dependent. From the definition of $\rho$ we see that if $|e_i|<c$, least square is performed in \eqref{min_robust} and hence the model respects the additive normal Gaussian noise assumption in \eqref{ols_eq}. When $|e_i|\ge c$, $e_i$ is no longer considered normal Gaussian but assumed to follow a Laplacian distribution of the form
$$
P(e_i) = f_c(e_i) = \left[C(c) e^{\frac{1}{2}c^2}\right]e^{-c|e_i|},
$$ 
where $C(c)$ is a constant such that $\int_{-\infty}^\infty f_c(x)\ dt = 1$. Note that the Laplacian distribution has a wider tail than a Gaussian distribution and hence allows for the existence of large $|e_i|$ better than the Gaussian distribution. Another popular robust choice for $\rho$ is the least absolute deviation \cite{bassett1978asymptotic} which amounts to having $\rho(e_i) = |e_i|$, which also follows a Laplacian distribution. 

The standard LASSO \cite{tibshirani1996regression} amounts to learning a sparse parameter vector $a$ via minimizing 
\begin{equation*}
\min_{a} \left\{ \frac{1}{2\sigma^2} \|y-Xa\|_2^2 + \mu \|a\|_1\right\}.
\end{equation*}
This model is still sensitive to outliers. A modification to this model introduces an extra variable $z$ representing outliers (see \cite{nguyen2013robust} and references there in):
\begin{equation*}
\min_{a,z} \left\{\frac{1}{2\sigma^2} \|Xa - y - z\|_2^2 + \mu\|a\|_1 + \lambda \|z\|_1 \right\}.
\end{equation*}
Suppose $u_i = f(y_{i-p},\cdots, y_{i-1})$. In \cite{mateos2012robust}, a robust nonparametric model is proposed:
\begin{equation*}
\min_{f,z}\left\{\sum_{i=1}^N |y_i-u_i-z|^2 + \mu \|f\|_{\H} + \lambda \|z\|_1 \right\},
\end{equation*}
where $\H$ is a reproducing kernel Hilbert space (RKHS) which includes Sobolev spaces. A slightly different approach is proposed in \cite{ganti2015learning}, where each $u_i$ is given by $f((Xa)_i)$, and both $a$ and $f$ are learned via solving
\begin{equation*}
\inf_{a,f} \left\{\sum_{i=1}^N |y_i - u_i|^2 +  \lambda \|a\|_1\right\}.
\end{equation*}

 In connection with the proposed model \eqref{min_prob_eq}, take $r = 1$ and let $u_i = f(\{y_j\}_{j\le i-1})= \sum_{j=1}^p x_{ij} a_j$ with $x_{ij} = y_{i-j}$ and
$$
P(y_i|u_i,\theta_i) = P(y_i|u_i,\sigma) = \frac{1}{\sigma\sqrt{2\pi}} e^{-\frac{|y_i-u_i|^2}{2\sigma^2}}
$$
for some fixed $\alpha>0$. Note that $f$ is completely determined by $a\in \R^p$ and estimating the parameter vector $a$ with the LASSO prior amounts to minimizing
\begin{equation*}
\min_{a,y} \left\{\frac{1}{2\sigma^2}\|y-Xa\|_2^2 +\mu \|a\|_1 +  \lambda\sum_{i\in D} |y_i-\tilde{y}_i|\right\}.
\end{equation*}
Here, the sparse vector representing outliers is $\{y_i-\tilde{y}_i\}_{i\in D}$. Note that the proposed method \eqref{min_prob_eq} is much more general to accommodate other types of noise in the data that is not additive (multiplicative Gaussian, Poisson, negative binomial, etc.)

There are quite a few existing methods on estimating the parameters in the presence of missing data, and from a high level perspective, they align with the following two approaches. 

The first approach iteratively imputes missing and unobserved data in some manner and then use the  imputed and observed data to estimate the parameters. These methods include mean imputation, expectation maximization (EM) \cite{dempster1977maximum}, multiple imputation \cite{rubin1976inference}, among others. See \cite{newman2003longitudinal, peng2006advance, horton2007much} for a survey of some of these methods. Matrix completion \cite{candes2009exact, candes2010power} is a form of imputation where missing entries in the data matrix are imputed under the assumption that the data matrix has low rank. The proposed imputation performed in Algorithm \ref{alg1} for solving \eqref{min_prob_eq} follows along the line the iterative approach of the EM method \cite{dempster1977maximum}; but instead of maximizing the expectation we maximize the likelihood.

 The second approach either doesn't impute or only imputes the necessary missing entries that the observed entries depend on. For instance in full information maximization likelihood (FIML) method \cite{finkbeiner1979estimation}, only complete data points are used as inputs to estimate the parameters. Suppose we are only interested in estimating the constant $a_0$ with $p=q=0$, then all observed data points are complete. However, for $p>0$ and suppose one in every consecutive $p$ points are missing then the set of complete data points is empty and hence the FIML method is not applicable. The non-negative definite covariance method \cite{loh2012high, choi2013investigation} only considers observed data points as inputs as opposed to complete data points. Here the observed data points may depend on the missing data, however this method sets this dependency to zero, that is imposing zero boundary conditions on the unobserved entries for which some of the observed entries may depend on (see section 3.2.1 in \cite{choi2013investigation}.) Further modification was introduced to acquire non-negative definite condition for the covariance matrix. This second approach can also be applied to our problem and it is more appropriate when the missing entries are systematic as oppose to random. To motivate the problem, consider the case where the correlation function is a constant, that is $u_i = c$ for all $i$. Assuming, that the observed series $\widetilde{y}_D$ has no outliers, then $c$ can be approximated as the mean $\frac{1}{|D|}\sum_{i\in D} \tilde{y}_i$ and there is no need to impute $y_i$ for all $i\notin D$. Similarly, suppose $D = \{k,\cdots, N\}$ and $u_i = f(y_{i-1})$, then it is only necessary to impute $y_{k-1}$ as opposed to $y_j$ for all $j<k$. One can view $y_{k-1}$ as an unknown boundary condition. Thus in this second approach, the minimization problem becomes
\begin{equation}\label{min_prob_eq2}
\begin{split}
\min_{f,y_{\bar{D}},u_{D'},\theta}\Big\{
J(y_{\bar{D}},u_{D'},f,\theta) &=  -\sum_{i\in D} \log(P(y_i|u_i,\theta_i)) - \log(P(f)\\
& + \lambda \sum_{i\in D}|y_i - \tilde{y}_i|^r
\Big\},
\end{split}
\end{equation}
where $\bar{D}$ consists of $D$ and all of the indices $j$ such that $u_i$ depends on $y_j$ for some $i\in D$, and $D'$ consists of all the indices $k$  such that $u_i$ depends on $u_k$ for some $i\in D$. Note that the difference between \eqref{min_prob_eq} and \eqref{min_prob_eq2} is that in \eqref{min_prob_eq2} the sum is only over observed indices as opposed to all indices (observed and unobserved.) It would be interesting to compare this approach with the proposed method \eqref{min_prob_eq} and we leave this for a future work.

\section{Poisson Log-Linear Model}\label{plar_sec}

In this section, we consider the Poisson distribution for the conditional uncertainty condition  and using a log-linear correlation function to model $f$ in \eqref{min_prob_eq}. In particular, we suppose that $y_i$ is only conditioned on $u_i$ and that it follows the Poisson distribution, 
\begin{equation}\label{p_eq}
P(y_i|u_i) = \frac{u_i^{y_i} e^{-u_i}}{\Gamma(y_i+1)},
\end{equation}
where $\Gamma(y_i+1) = y_i!$ whenever $y_i$ is a nonnegative integer. Note that \eqref{p_eq} is defined for all $y\ge 0$. We consider $f$ to be a log-linear correlation function satisfying
\begin{equation}\label{u_eq}
\begin{split}
\log(u_i + 1) = a_0 + \sum_{k=1}^p a_k \log(y_{i-k}+1)+ \sum_{k=1}^q b_k \log(u_{i-k}+1),
\end{split}
\end{equation}
 for some $p,q\ge 0$ with the constraint $u_i\ge 0$. In other words,
$$
u_i = f\left(\{y_j\}_{j\le i-1},\{u_j\}_{j\le i-1}\right) = \max(\exp(\log(u_i+1)) - 1,0), 
$$
where $\log(u_i+1)$ is defined as in \eqref{u_eq}. Since $f$ is completely determined by $a_0$, $a$ and $b$, we see that the series $u$ is completely determined by $a_0,a,b$ and the series $y$. 

The prior on $f$ is now given by
$$
P(f) = P(a_0)\left[\prod_{k=1}^p P(a_k)\right]\left[\prod_{k=1}^q P(b_k)\right].
$$
Here we assume all the parameters are independent from each other, and follow a family of exponential probability distributions
\begin{equation}\label{prior_f_eq}
f_{\mu,s}(x) = C(\mu,s)e^{-\mu|x|^{s}},\mbox{ for some } 0<s\le 1,
\end{equation} 
where $C(\mu,s)$ is chosen such that $\int_{-\infty}^\infty f_{\mu,s}(x)\ dx = 1$. $s = 1$ corresponds to the LASSO constraint \cite{tibshirani1996regression} and $s\in (0,1)$ corresponds to the Bayesian bridge constraint proposed in \cite{frank1993statistical, polson2014bayesian}. 

Combining \eqref{p_eq}-\eqref{prior_f_eq} into \eqref{min_prob_eq}, the proposed variational problem becomes

\begin{equation}\label{pois_min_eq}
\begin{split}
&\min_{a_0,a,b,y} J(a_0,a,b,y),\mbox{ where }\\
J(a_0,a,b,y) &= \sum_{i=1}^N \left[u_i - y_i\log(u_i) + \log(\Gamma(y_i+1)) \right]+ \lambda\sum_{i\in D}|y_i - \tilde{y}_i|^{r}\\
&+ \mu \left(\sum_{k=1}^p |a_k|^{s} + \sum_{k=1}^q |b_k|^{s}\right),
\end{split}
\end{equation}
for some $0<r,s\le 1$, and $\mu,\lambda>0$.

\begin{remark}
It is possible to minimize the energy in \eqref{pois_min_eq} over the set of nonnegative integer-valued series $y$. However, this set is not convex. To overcome this non-convexity we extend \eqref{p_eq} to all nonnegative real-valued series $y$ and therefore use $\Gamma(y_i+1)$ as oppose to $y_i!$. We remark that this extension is not the continuous version of the Poisson distribution proposed in \cite{marsaglia1986incomplete, ilienko2013continuous}. Given $\lambda\ge 0$, the cumulative probability distribution for the continuous Poisson is defined as $F_\lambda(t) = 0$ if $t=0$ and
$$
F_\lambda(t) = \frac{\Gamma(t,\lambda)}{\Gamma(t)},\mbox{ for } t>0.
$$
Here
$$
\Gamma(t,\lambda) := \int_\lambda^\infty e^{-\tau} \tau^{t-1}\ d\tau
$$
is the incomplete Gamma function. Let the probability distribution function $f_\lambda$ be defined such that
$$
F_\lambda(t) = \int_0^t f_\lambda(\tau)\ d\tau.
$$     
It can be shown that for all $t\ge 0$,
$$
\int_t^{t+1} f_\lambda(\tau)\ d\tau = \frac{e^{-\lambda}\lambda^t}{\Gamma(t+1)}.
$$
Thus instead of \eqref{p_eq}, one can use $P(y_i|u_i) = f_{u_i}(y_i)$ or
$$
P(y_i|u_i) =   \frac{e^{-u_i}u_i^{\lfloor y_i\rfloor}}{\Gamma(\lfloor y_i\rfloor+1)},
$$
where $\lfloor t\rfloor$ is the largest integer that is less than or equal to $t$. 
\end{remark}

\section{Proximal Mapping for $\ell^r$, $0<r<1$}\label{prox_p_sec}

One of the ingredients for computing \eqref{pois_min_eq} is to solve a subproblem
\begin{equation}\label{lp_l2}
\shrink_r(t',\mu) =J_{\d E_r}(t')  := \argmin_{t\in \R} \left\{E_{r}(t) = \mu |t|^r + \frac{1}{2}|t-t'|^2\right\},
\end{equation}
for some $0<r\le 1$, $\mu>0$ and $t'\in \R$.  For $r=1$, it was shown in \cite{tibshirani1996regression} that $J_{\d E_1}(t')$ is given by 
\begin{equation}\label{shrink_1}
J_{\d E_1}(t') = \shrink_1(t',\mu) :=  \sign(t')(|t'|-\mu)_+,
\end{equation}
where $x_+ = \max(x,0)$. For $0<r<1$, a common approach for solving \eqref{lp_l2} is to consider a regularized version via
$$
\min_{t\in \R} \left\{H_{r}(t) = \mu |t+\epsilon|^r + \frac{1}{2}|t-t'|^2\right\},
$$
for some $\epsilon>0$. The minimizer for $H_r(t)$ is approximated  by $t_\infty = \lim_{n\rightarrow \infty} t_n$ where
$$
t_{n+1} = t_n - dt\frac{d H_r}{d t}(t_n),
$$
with some initial guess $t_0$. Since $H_r(t)$ is nonconvex, $t_\infty$ may not be a global minimizer. In remark \ref{shrink_rem} below we show that for a range of $\mu$, with the initial guess $t_0=t'$, the above iteration will converge to $t_\infty$ that is not a global minimizer.

There are explicit formula for $\shrink_r$ when $r=1/2$ or $r=2/3$ \cite{zhongben2012regularization}, but for general $r\in(0,1)$, no closed-form expression for $\shrink_r$ exists. In \cite{nie2012robust}, Nie-Etal proposed a method for solving $\shrink_r$ via computing a zero of a strictly convex function using Newton method. For completeness, we present here a method similar to the one proposed in \cite{nie2012robust}. 

Recall \eqref{lp_l2} and for simplicity assume $t'\ge 0$ and denote by $t^*$ the global minimizer for $E_r(t)$. First, note that $t^*\ge 0$. Now, for $t>0$ we get
$$
\frac{d}{d t}E_r(t) = \mu r t^{r-1} + (t-t').
$$
Define
$$
g_r(t) = \mu r - t' t^{1-r} + t^{2-r},\mbox{ for } t \ge 0.
$$
Note also that for $t>0$, $g_r(t) = t^{1-r} \frac{d E_r}{dt}(t)$ which shows that $g_r(t) = 0$ if and only if $\frac{d }{d t}E_r(t)=0$. Suppose $t^*>0$. This implies that $t^*$ must satisfy
\begin{equation}\label{lr_deri_eq}
g_r(t^*) = \mu r -t' (t^*)^{1-r}+ (t^*)^{2-r} = 0.
\end{equation}
On $(0,\infty)$, we get
$$
g_r'(t) = -(1-r) t' t^{-r} + (2-r)t^{1-r},
$$
and
$$
g_r''(t) = (1-r)r t' t^{-r-1} + (2-r)(1-r)t^{-r},
$$
which is strictly greater than $0$. This implies $g_r$ is strictly convex on $(0,\infty)$ and achieving its minimal value at $t_0 = \frac{1-r}{2-r}t'$. Moreover, we have
$$
g_r(0) = g_r(t') = \mu r>0,\mbox{ and } 
$$
$$
g_r(t_0) = \mu r +  \left(t_0- t'\right)t_0^{1-r} < \mu r. 
$$ 
This implies the followings:
\begin{enumerate}
\item If $g_r(t_0) > 0$, that is
\begin{equation}\label{case1_eq}
\mu > \frac{1}{r}  \left(t' - t_0\right) t_0^{1-r},
\end{equation}
then $g_r$ has no zeros on $(0,\infty)$. This shows that a global minimizer $t^*>0$ does not exist. 
\item If $g_r(t_0) < 0$, that is
\begin{equation}\label{case2_eq}
\mu < \frac{1}{r}  \left(t' -t_0\right) t_0^{1-r},
\end{equation}
then $g_r$ has exactly two zeros $t_1\in (0,t_0)$ and $t_2\in (t_0,t')$. Since $0<r<1$ and $\mu>0$, the function $E_r(t)$ is strictly increasing near zero. This implies that the zero $t_1$ is not a local minimizer for $E_r$, and hence $t^* = t_2$. 
\item If $g_r(t_0) = 0$, that is
\begin{equation}\label{case3_eq}
\mu = \frac{1}{r}  \left(t' - t_0\right) t_0^{1-r},
\end{equation}
then $g_r$ has exactly one zero at $t_0$. Since $t_0$ is the first zero of $g_r$ and $E_r$ is strictly increasing near $0$, we get that  a global minimizer $t^*>0$ does not exist.
\end{enumerate}

\begin{remark}\label{shrink_rem}
Cases \eqref{case1_eq} and \eqref{case3_eq} correspond to having $E_r(t)$ strictly increasing on $[0,\infty)$, and hence $t^*=0$ is the global minimizer.  Note in the case \eqref{case3_eq}, $t_0$ is a saddle point of $E_r$. As for the case \eqref{case2_eq}, $E_r(t)$ has one positive local minimizer at $t_2$ and therefore the global minimizer $t^* = \argmin_{t\in \{0,t_2\}} E_r(t)$. It is possible that $E_r(0) = E_r(t_2)$ for some $\mu>0$. In this case there is no uniqueness to the global minimizer of $E_r$.
\end{remark}

Figure \ref{prox_p_fig} shows the plots of $E_r(t)$ and the corresponding $g_r(t)$ with $r=\frac{1}{2}$ and $t^\prime = 5$ for various choices of $\mu$. The green lines correspond to the value $E_r(0)$ in $(i)$ and $0$ in $(ii)$. Let $\mu_0 =  \frac{1}{r}  \left(t' - t_0\right) t_0^{1-r}$. In the cases where $\mu = 2\mu_0, \mu_0,\frac{3}{4}\mu_0$, the global minimizer for $E_r$ is $0$. However, for $\mu = \frac{1}{4}\mu_0$, the second zero of the corresponding $g_r$ is the global minimizer for $E_r$.

\begin{figure}
\begin{center}
$\underset{(i)}{\includegraphics[scale=0.45]{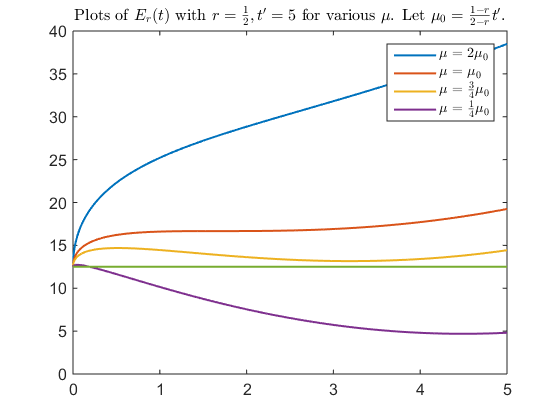}}$
$\underset{(ii)}{\includegraphics[scale=0.45]{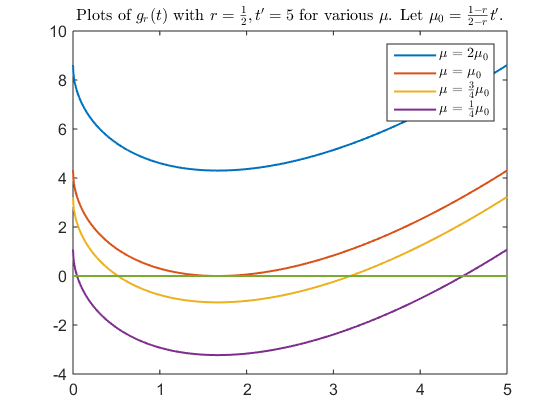}}$
\caption{Plots of $E_r(t)$ and the corresponding $g_r(t)$ with $r=\frac{1}{2}$ and $t^\prime = 5$ for various choices of $\mu$. The green lines correspond to the value $E_r(0)$ in $(i)$ and $0$ in $(ii)$.}
\label{prox_p_fig}
\end{center}
\end{figure}

From the above remark, the shrinkage operator in \eqref{lp_l2} for some $t'\ge 0$ is given by
\begin{equation}\label{shrink_pos_eq}
\shrink_r(t',\mu) =  \left\{
\begin{matrix}
0 & \mbox{ if }& g_r(t_0) \ge 0,\\
\argmin_{t\in \{0,t_2\}} \{E_r(t)\} & \mbox{ if } & g_r(t_0) < 0.
\end{matrix}
\right.
\end{equation}
Computing for the second zero $t_2\in (t_0,t')$ of $g_r$ (assuming $g_r(t_0) < 0$) is fast and straightforward since $g_r$ is strictly convex. For instance, one can use the Newton method as follow:
\begin{algorithm}
Newton method for computing the second zero of $g_r$ (assuming $t'\ge 0$ and $g_r(t_0) < 0$):
\begin{enumerate}
\item Set $t_0  = t'$, $t_1 = t_0 - g_r(t_0)/g_r'(t_0)$ and $\epsilon = $ small.
\item while $|t_m - t_{m-1}| > \epsilon$.
\begin{itemize}
\item $t_{m+1} = t_m - g_r(t_m)/g_r'(t_m)$.
\end{itemize}
\item end while. 
\end{enumerate}
\end{algorithm} 
In our numerical simulation, the above algorithm converges in $\le 5$ Newton iterations with $\epsilon = 1e$-$6$.

For general $t'\in \R$, we have 
\begin{equation}
\shrink_r(t',\mu) = \sign(t') \shrink_r(|t'|,\mu)
\end{equation}
 where $\shrink_r(|t'|,\mu)$ is given in \eqref{shrink_pos_eq}.

\begin{remark}
$E_r(t)$ can also be defined for $r=0$. In this case, we get $E_0(0) = \frac{1}{2}(t')^2$ and $E_0(t) = \mu + \frac{1}{2} (t-t')^2$ for $t\neq 0$. This implies for $t'\neq 0$, $\frac{dE_0}{dt}(t) = 0$ whenever $t = t'$. Thus we obtain the following shrinkage operator (hard thresholding) 
\begin{equation}
\shrink_0(t',\mu) = \left\{
\begin{matrix}
0 &\mbox{ if } & \frac{1}{2}(t')^2 < \mu\\
t' &\mbox{ if } & \frac{1}{2}(t')^2 \ge \mu.
\end{matrix}
\right.
\end{equation}
Note that if $E_0(0) = \frac{1}{2}(t')^2 = \mu = E_0(t')$ then there is no uniqueness of minimizer. Here we choose $t'$ to be the minimizer but choosing $0$ is also appropriate.
\end{remark}

\section{Numerical Implementation}\label{numeric_sec}

There are numerous numerical methods that can be used to solve a minimizer for \eqref{pois_min_eq} (see \cite{goldstein2014fast} and references there in.) We mention in particular the FISTA algorithm \cite{beck2009fast} which provides a global rate of convergence when the minimizing energy is convex. The functional in \eqref{pois_min_eq} is related blind-deconvolution which is jointly nonconvex even in the case when $r=s=1$. When both $r$ and $s$ are rationals in $(0,1)$, numerical schemes such as PALM \cite{bolte2014proximal} or Block Prox-Linear Method \cite{xu2014globally} provides global convergence. Numerical schemes FISTA and PALM are described in algorithms \ref{fista}-\ref{palm}. Algorithm \ref{alg1} is a combination of the two where we apply the time-step updating criteria from FISTA to the proximal alternating scheme in PALM. We will make comparisons between PALM and Algorithm \ref{alg1} via numerical simulations.

We rewrite the energy from \eqref{pois_min_eq} as 
$$
J(a_0,a,b,y) = H(a_0,a,b,y) + G_1(a) + G_2(b) + G_3(y_D-\tilde{y}_D)
$$ 
where
\begin{eqnarray*}
H(a_0,a,b,y) &=& \sum_{i=1}^N\left[u_i - y_i\log(u_i) + \log(\Gamma(y_i+1))\right] \\
G_1(a)&=& \mu \sum_{k=1}^p |a_k|^{s},\ G_2(b) = \mu\sum_{k=1}^q |b_k|^{s},\\
G_3(y_D-\widetilde{y}_D) &=& \lambda \sum_{i\in D} |y_i - \tilde{y}_i|^{r}.
\end{eqnarray*}

Let 
$$
\nabla H = \left(\nabla_{a_0} H,\nabla_{a} H, \nabla_{b} H, \nabla_{y} H\right),
$$
where 
\begin{eqnarray*}
\nabla_{a_0}H := \frac{\d H}{\d a_0},& & \nabla_{a} H := \left(\frac{\d H}{\d a_k}\right)_{k=1}^p,\\
\nabla_{b} H := \left(\frac{\d H}{\d b_k}\right)_{k=1}^q, & &\nabla_{y} H := \left(\frac{\d H}{\d y_i}\right)_{k=1}^N.
\end{eqnarray*}
 The proximal mappings for $G_i$'s are defined as:
\begin{eqnarray*}
J_{\d G_1}(a') &:=& \argmin_{a}\left\{G_1(a) + \frac{1}{2} \|a-a'\|_2^2\right\} = (\shrink_s(a'_k,\mu))_{k=1}^p,\\
J_{\d G_2}(b') &:=& \argmin_{b}\left\{G_2(b) + \frac{1}{2} \|b-b'\|_2^2\right\} = (\shrink_s(b'_k,\mu))_{k=1}^q, \\
J_{\d G_3}(x_D') &:=& \argmin_{x_D}\left\{G_3(x_D) + \frac{1}{2} \|x_D-x_D'\|_2^2\right\}\\
&=& \left(\shrink_r(x_i',\lambda)\right)_{i\in D},
\end{eqnarray*}
where $x$ is a vector in $\R^N$ such that $x_D = y_D - \widetilde{y}_D$. In the followings, denote by $\P_D(y)$ and $\P_{D^c}(y)$ the projection of $y\in \R^N$ on to $D$ and its complement $D^c$, respectively.

\begin{algorithm}\label{fista}
FISTA \cite{beck2009fast} applying to \eqref{pois_min_eq}.
\begin{itemize}
\item Initialize: $c_{0,2} = a_{0,1}$, $c_2 = a_1$, $d_2 = b_1$, $z_2 = y_1$, $\alpha_2=1$, $\tau= $ small enough and $\epsilon > 0$ (tolerance.)
\item Do
\begin{enumerate}
\item $\alpha_{m+1} = (1 + \sqrt{1 + 4\alpha_m^2})/2$.
\item $a_{0,m} = c_{0,m} - \tau \nabla_{a_0} H(c_{0,m},c_m,d_m,z_m)$. 
\item $c_{0,m+1} = a_{0,m} + \frac{\alpha_m-1}{\alpha_{m+1}}(a_{0,m} - a_{0,m-1})$.
\item $a_m= J_{\tau\d G_1} \left[c_m - \tau \nabla_a H(c_{0,m},c_m,d_m,z_m)\right]$.
\item $c_{m+1} =  a_{m} + \frac{\alpha_m-1}{\alpha_{m+1}}(a_{m} - a_{m-1})$.
\item $b_m= J_{\tau\d G_2} \left[d_m - \tau \nabla_b H(c_{0,m},c_{m},d_m,z_m)\right]$.
\item $d_{m+1} =  b_{m} + \frac{\alpha_m-1}{\alpha_{m+1}}(b_{m} - b_{m-1})$.
\item $\P_{D^c}(y_m) = \P_{D^c}\left(z_m - \tau \nabla_y H(c_{0,m},c_{m},d_{m},z_m)\right)$.
\item  $\P_{D}(y_m) = J_{\tau\d G_3} \left[\P_{D}\left(z_m - \tau \nabla_y H(c_{0,m},c_{m},d_{m},z_m)\right) - \widetilde{y}_D\right] + \widetilde{y}_D$.  
\item $z_{m+1} = y_m + \frac{\alpha_m-1}{\alpha_{m+1}} (y_m - y_{m-1})$. 
\end{enumerate}
\item while $|J_{m} - J_{m-1}| > \epsilon$.
\item Set $(a_0^*,a^*,b^*,y^*) = (a_{0,m},a_m,b_m,y_m)$.
\end{itemize}
\end{algorithm}

\begin{algorithm}\label{palm}
PALM \cite{bolte2014proximal} applying to \eqref{pois_min_eq}.
\begin{itemize}
\item Initialize: $a_{0,1}$, $a_1$, $b_1$, $y_1$, $\tau= $ small enough and $\epsilon > 0$ (tolerance.)
\item Do
\begin{enumerate}
\item $a_{0,m+1} = a_{0,m} - \tau \nabla_{a_0} H(a_{0,m},a_m,b_m,y_m)$. 
\item $a_{m+1}= J_{\tau\d G_1} \left[a_m - \tau \nabla_a H(a_{0,m+1},a_m,b_m,y_m)\right]$.
\item $b_{m+1}= J_{\tau\d G_2} \left[b_m - \tau \nabla_b H(a_{0,m+1},a_{m+1},b_m,y_m)\right]$.
\item $\P_{D^c}(y_{m+1}) = \P_{D^c}\left(y_m - \tau \nabla_y H(a_{0,m+1},a_{m+1},b_{m+1},y_m)\right)$.
\item  $\P_{D}(y_{m+1}) = J_{\tau\d G_3} \left[\P_{D}\left(y_m - \tau \nabla_y H(a_{0,m+1},a_{m+1},b_{m+1},y_m)\right) - \widetilde{y}_D\right] + \widetilde{y}_D$.  
\end{enumerate}
\item while $|J_{m+1} - J_{m}| > \epsilon$.
\item Set $(a_0^*,a^*,b^*,y^*) = (a_{0,m},a_m,b_m,y_m)$.
\end{itemize}
\end{algorithm}

\begin{algorithm}\label{alg1}
Computing an optimal $(a_0^*,a^*,b^*,y^*)$ for \eqref{pois_min_eq}.
\begin{itemize}
\item Initialize: $c_{0,2} = a_{0,1}$, $c_2 = a_1$, $d_2 = b_1$, $z_2 = y_1$, $\alpha_2=1$, $\tau= $ small enough and $\epsilon > 0$ (tolerance.)
\item Do
\begin{enumerate}
\item $\alpha_{m+1} = (1 + \sqrt{1 + 4\alpha_m^2})/2$.
\item $a_{0,m} = c_{0,m} - \tau \nabla_{a_0} H(c_{0,m},c_m,d_m,z_m)$. 
\item $c_{0,m+1} = a_{0,m} + \frac{\alpha_m-1}{\alpha_{m+1}}(a_{0,m} - a_{0,m-1})$.
\item $a_m= J_{\tau\d G_1} \left[c_m - \tau \nabla_a H(c_{0,m+1},c_m,d_m,z_m)\right]$.
\item $c_{m+1} =  a_{m} + \frac{\alpha_m-1}{\alpha_{m+1}}(a_{m} - a_{m-1})$.
\item $b_m= J_{\tau\d G_2} \left[d_m - \tau \nabla_b H(c_{0,m+1},c_{m+1},d_m,z_m)\right]$.
\item $d_{m+1} =  b_{m} + \frac{\alpha_m-1}{\alpha_{m+1}}(b_{m} - b_{m-1})$.
\item $\P_{D^c}(y_m) = \P_{D^c}\left(z_m - \tau \nabla_y H(c_{0,m+1},c_{m+1},d_{m+1},z_m)\right)$.
\item  $\P_{D}(y_m) = J_{\tau\d G_3} \left[\P_{D}\left(z_m - \tau \nabla_y H(c_{0,m+1},c_{m+1},d_{m+1},z_m)\right) - \widetilde{y}_D\right] + \widetilde{y}_D$.  
\item $z_{m+1} = y_m + \frac{\alpha_m-1}{\alpha_{m+1}} (y_m - y_{m-1})$. 
\end{enumerate}
\item while $|J_{m} - J_{m-1}| > \epsilon$.
\item Set $(a_0^*,a^*,b^*,y^*) = (a_{0,m},a_m,b_m,y_m)$.
\end{itemize}
\end{algorithm}

Recall,
\begin{eqnarray*}
u_i &=& \exp(\log(u_i+1)) - 1\\
&=& \exp\left(a_0 + \sum_{k=1}^p a_k \log(y_{i-k}+1)+ \sum_{k=1}^q b_k \log(u_{i-k}+1)\right) - 1.
\end{eqnarray*}
The differentials of $H$ with respect to its variables are as follows.
\begin{equation*}
\begin{split}
\frac{\d H}{\d a_0} = \sum_{i=1}^N \left[ \frac{\d u_i}{\d a_0} - \frac{y_i}{u_i}\frac{\d u_i}{\d a_0} \right],
\end{split}
\end{equation*}
where
\begin{equation*}
\begin{split}
&\frac{\d u_i}{\d a_0} = (u_i+1) \frac{\d}{\d a_0}(\log(u_i + 1)),\mbox{ and }\\
&\frac{\d}{\d a_0}(\log(u_i + 1)) = 1 + \sum_{k=1}^q b_k  \frac{\d}{\d a_0}(\log(u_{i-k} + 1)). 
\end{split}
\end{equation*}
Therefore,
\begin{equation}
\begin{split}
\frac{\d H}{\d a_0} = \sum_{i=1}^N \left[\frac{ u_i - y_i}{u_i}\right](u_i+1).
\end{split}
\end{equation}
We have
\begin{equation}
\frac{\d H}{\d a_k} = \sum_{i=1}^N \left[\frac{ u_i - y_i}{u_i}\right](u_i+1)\frac{\d}{\d a_k}(\log(u_i+1)),
\end{equation}
where 
$$
\frac{\d}{\d a_k}(\log(u_i+1)) = \log(y_{i-k}+1) + \sum_{\ell=1}^q b_\ell  \frac{\d}{\d a_k}(\log(u_{i-\ell} + 1)) 
$$
Similarly,
\begin{equation}
\frac{\d H}{\d b_k} = \sum_{i=1}^N \left[\frac{ u_i - y_i}{u_i}\right](u_i+1)\frac{\d}{\d b_k}(\log(u_i+1)),
\end{equation}
where 
$$
\frac{\d}{\d b_k}(\log(u_i+1)) = \log(u_{i-k}+1) + \sum_{\ell=1}^q b_\ell  \frac{\d}{\d b_k}(\log(u_{i-\ell} + 1)) 
$$
Lastly,
\begin{equation}
\begin{split}
\frac{\d H}{\d y_i} &= -\log(u_i) + \frac{\Gamma'(y_i+1)}{\Gamma(y_i+1)} \\
&+ \sum_{j=1}^N \left[\frac{u_j-y_j}{u_j}\right](u_j+1) \frac{\d}{\d y_i}(\log(u_j+1)),
\end{split}
\end{equation}
where
\begin{eqnarray*}
\frac{\d}{\d y_i}(\log(u_j+1)) &=& \sum_{k=1}^q b_k \frac{\d}{\d y_i}(\log(u_{j-k}+1)) \\
&+& 
\left\{
\begin{matrix}
\frac{a_{j-i}}{y_i+1} & \mbox{ if } 1\le j-i \le p\\
0 &\mbox{ otherwise.} &
\end{matrix}
\right.
\end{eqnarray*}

\section{Numerical results}\label{num_res}

In this section, we validate the model \eqref{pois_min_eq} with simulated data. We argue that the regularization on the parameters with $\ell^s$, $0<s\le 1$, and the sparsity constraint using $\ell^r$, $0<r\le1$, are all crucial and necessary in estimating the parameters accurately. Throughout this section, we simulate data according to conditions \eqref{p_eq} and \eqref{u_eq} iteratively using $p=6$, $q=0$ and the true parameters 
\begin{eqnarray*}
\tilde{a}_0 &=& 1,\\
\tilde{a} &=& (\tilde{a}_1,\tilde{a}_2,\tilde{a}_3,\tilde{a}_4,\tilde{a}_5) \\
&=& (0.25,-0.5,0, 0,-0.5,0.5),
\end{eqnarray*}
with the size of the series $N=1000$. In all the figures below, $\tilde{y}$ and $\tilde{u}$ are the simulated observed and true mean series using these parameters. We then apply partial series of $\tilde{y}$ to the model described in \eqref{pois_min_eq} to reconstruct the extended series $y$, its mean $u$ and the parameters $a_0$, $a$. Also plotted in these figures is the vector $D$ with the following interpretation: $D_i = 1$ implies $\tilde{y}_i$ is observed and $D_i=0$ implies $\tilde{y}_i$ is unobserved.

\begin{example}
Figure \ref{ex0_fig1} shows a comparison of performance between Algorithms \ref{palm} and \ref{alg1} for 100 simulations. For each simulation, only $75\%$ of entries are observed and among these entries $2.5\%$ are contaminated.  In both algorithms, the parameters used are $\lambda = 5, r = 0.5,\mu = 30$ and $s = 1$. For Algorithm \ref{palm} we use $\tau = 1e$-$4$, and for Algorithm \ref{alg1} we use $\tau = 1e$-$5$. Even with a smaller $\tau$, Algorithm \ref{alg1} converges in 800 iterations on average, where as it takes on average 10,000 iterations for Algorithm \ref{palm} to converge. From the plots we observe that both algorithms provide similar statistics on the estimated parameters.
\end{example}

\begin{figure}
\begin{center}
\includegraphics[scale=0.45]{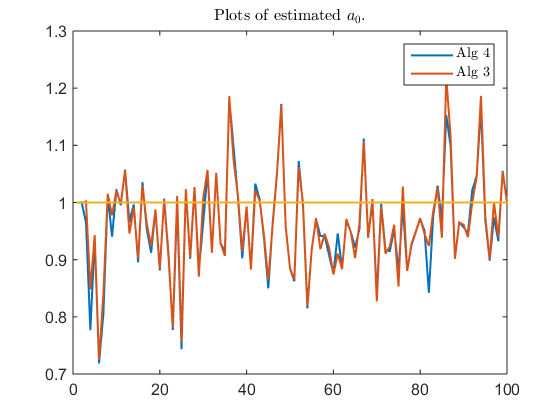}\\
\includegraphics[scale=0.45]{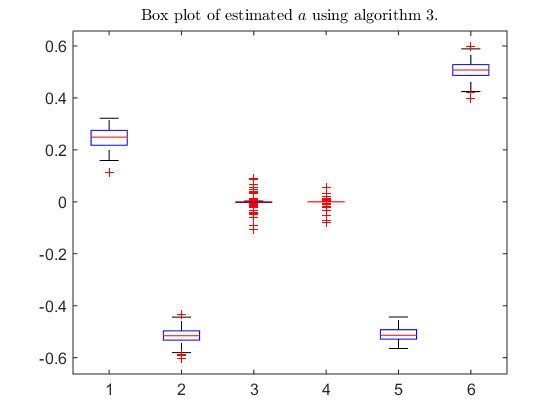}
\includegraphics[scale=0.45]{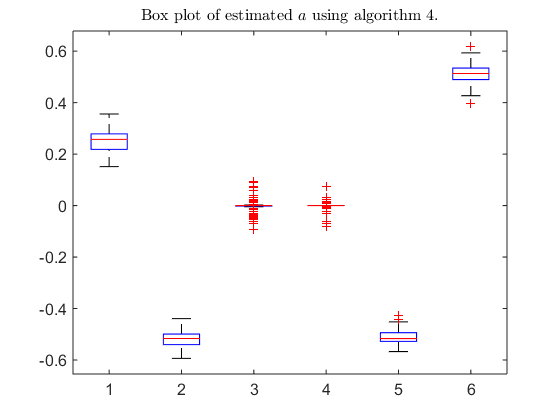}
\caption{Plots of estimated parameters $a_0$ and $a$ using Algorithm \ref{palm} and \ref{alg1} for 100 simulations when $75\%$ of the data is observed and among the observed entries $2.5\%$ are contaminated. In both algorithms, the parameters used are $r = 1/2$, $\lambda = 5$, $s = 1.0$ and $\mu = 30$. From the plots we observe that both algorithms provide similar statistics on the reconstructed parameters. However, Algorithm \ref{alg1} converges in about $800$ iterations where as it takes about $10,000$ iterations for Algorithm \ref{palm} to converge.}
\label{ex0_fig1}
\end{center}
\end{figure}

\begin{example}\label{ex1}
In this example, using Algorithm \ref{alg1} we compare the results of estimating the parameters $a_0$ and $a$ with and without using regularization on $a$. Figures \ref{ex1_fig1}-\ref{ex1_fig3} show the plots of the estimated $a_0$ and $a$ for three different amounts of observations (100\%, 75\% and 50\%). In all these cases we use $r=1/2, \lambda = 5$ and $s=1$. The values of $\mu$ changes depending on the amount of missing data. As the amount of missing data increases, the parameter estimation performance deteriorates when no regularization on $a$ is used (i.e. $\mu=0$). However, with $\ell^s$ regularization on $a$, the estimated parameters are much closer to the true values. 
\end{example}

\begin{figure}
\begin{center}
\includegraphics[scale=0.45]{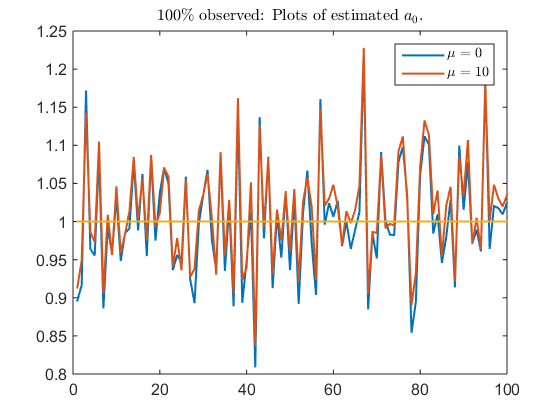}\\
\includegraphics[scale=0.45]{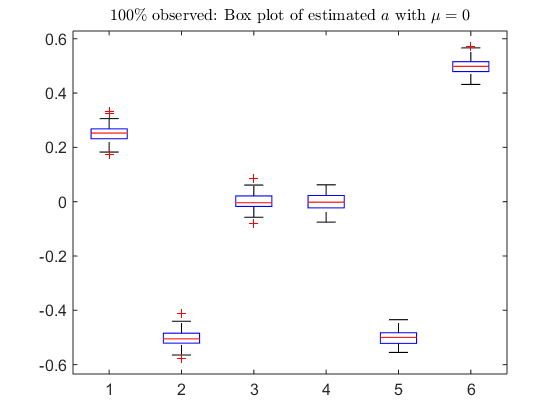}
\includegraphics[scale=0.45]{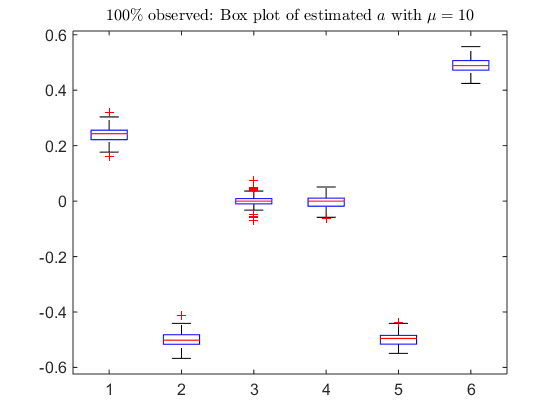}
\caption{Plots of estimated parameters $a_0$ and $a$ for 100 simulations when $100\%$ of the data is observed with $\mu = 0$ (left) and $\mu=10$ (right). The other parameters are $r = 1/2$, $\lambda = 5$, $s = 1.0$. From the plots, we observe that there is no significant difference between having $\mu=0$ (without regularization on $a$) and $\mu = 10$ (with $\ell^s$ regularization on $a$).}
\label{ex1_fig1}
\end{center}
\end{figure}

\begin{figure}
\begin{center}
\includegraphics[scale=0.45]{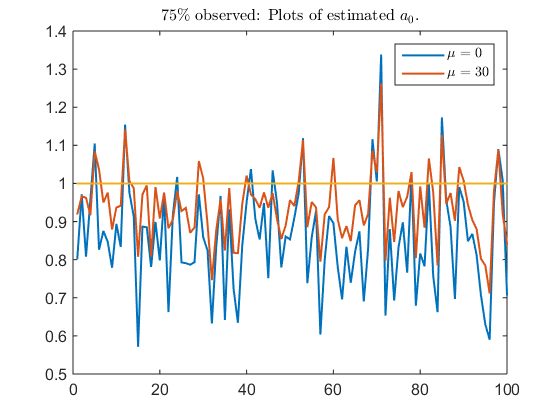}\\
\includegraphics[scale=0.45]{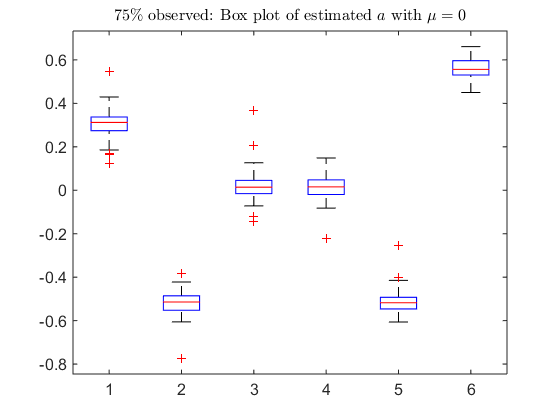}
\includegraphics[scale=0.45]{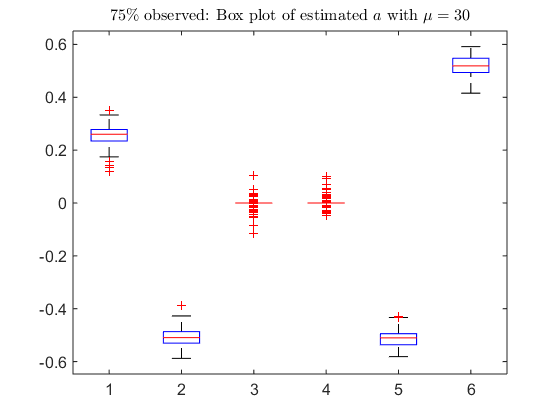}
\caption{Plots of estimated parameters $a_0$ and $a$ for 100 simulations when $75\%$ of the data is observed randomly with $\mu = 0$ (left) and $\mu=30$ (right). The other parameters are $r = 1/2$, $\lambda = 5$, $s = 1.0$. From the plots, we observe some improvements for having $\mu = 30$ over having $\mu=0$. The box plots also show that there is less variation for the estimated parameter $a$ when using regularization with $\mu = 30$.}
\label{ex1_fig2}
\end{center}
\end{figure}

\begin{figure}
\begin{center}
\includegraphics[scale=0.45]{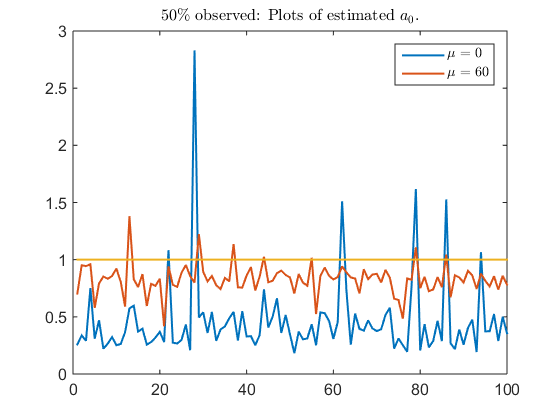}\\
\includegraphics[scale=0.45]{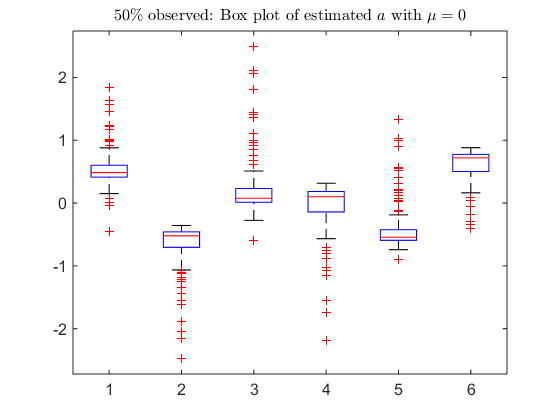}
\includegraphics[scale=0.45]{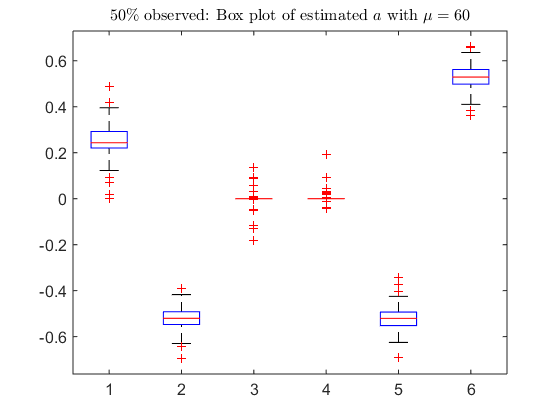}
\caption{Plots of estimated parameters $a_0$ and $a$ for 100 simulations when $50\%$ of the data is observed randomly with $\mu = 0$ (left) and $\mu=60$ (right). The other parameters are $r = 1/2$, $\lambda = 5$, $s = 1.0$. From the plots, we clearly observe significant improvements for having $\mu = 60$ (with $\ell^s$ regularization on $a$) over having $\mu=0$ (without using regularization). For $\mu = 0$, the box plot on the left shows that the estimated values for $a$ deviate significantly from the true values.}
\label{ex1_fig3}
\end{center}
\end{figure}

\begin{example}
In this example, using Algorithm \ref{alg1} we show the significance of having the sparsity constraint $\sum_{i\in D} |y_i-\widetilde{y}_i|^r$ in the model \eqref{pois_min_eq} to desensitize anomalies and outliers for obtaining a more accurate parameter estimation.  Here we assume all data are observed, that is $D = \{1,\cdots, N\}$. For each simulation, we randomly select a certain percentage (1\%, 5\% or 10\%) of the data and replace them with some anomalous value (here we pick $20$ as an example). See Figure \ref{ex2_fig0} for an example showing the original series and a contaminated version. To test the significance of the sparsity constraint term, we consider two scenarios:
\begin{enumerate} 
\item Assuming the observed data has no outliers and hence enforcing $y_i=\widetilde{y}_i$, $i\in D$. This amounts to picking $\lambda$ to be large, say $\lambda = 50$.
\item Assuming the observed data has outliers and hence allowing for $y_i\neq \widetilde{y}_i$, $i\in D$, whenever $\widetilde{y}_i$ is an anomaly. This amounts to picking $\lambda$ to be small, say $\lambda = 2$.
\end{enumerate}
In all cases the remaining parameters are: $r = 1/2, s = 1$ and $\mu = 10$. Figures \ref{ex2_fig1}-\ref{ex2_fig3} show the plots of the estimated parameters $a_0$ and $a$ for $100$ simulations with the amount of contaminations to be $1\%$, $5\%$ and $10\%$. We observe that enforcing $y_i=\widetilde{y}_i$ (that is $\lambda = 50$) greatly alters the reconstruction of $a_0$ and $a$ and this error increases as the amount of contamination increases. This enforcement causes an increase in the mean value (see Figures \ref{ex2_fig4}-\ref{ex2_fig6}) and a decrease in the absolute values of the correlation coefficients in the reconstructed time series (see Figures \ref{ex2_fig1}-\ref{ex2_fig3}.) Letting $\lambda = 2$, which is small, prevents the model from fitting $y_i$ to the anomalous $\widetilde{y}_i$. As a result, the estimated parameters $a_0$ and $a$ are much closer to the ground truth. The inaccuracy in parameter estimation effects the reconstruction of the mean series, and hence the prediction. The results are visibly seen in Figures \ref{ex2_fig4}-\ref{ex2_fig6}.
\end{example}

\begin{figure}
\begin{center}
\includegraphics[scale=0.45]{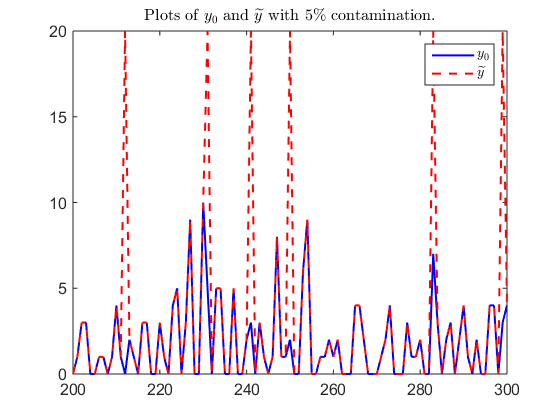}\\
\caption{The plots show the original series $y_0$ and the series $\widetilde{y}$ with $5\%$ contamination.}
\label{ex2_fig0}
\end{center}
\end{figure}

\begin{figure}
\begin{center}
\includegraphics[scale=0.45]{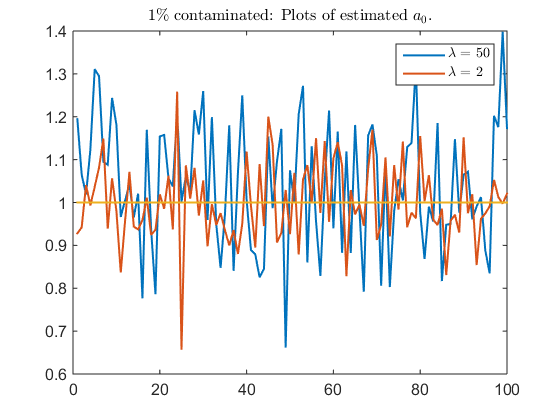}\\
\includegraphics[scale=0.45]{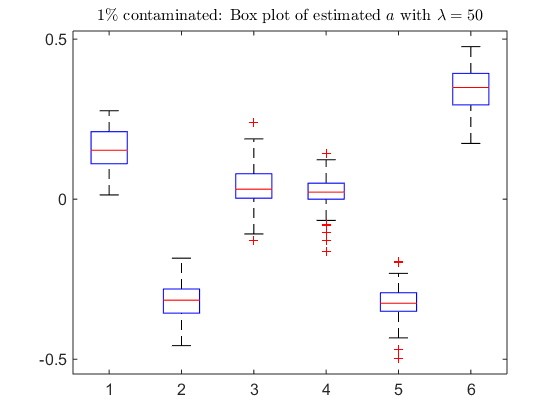}
\includegraphics[scale=0.45]{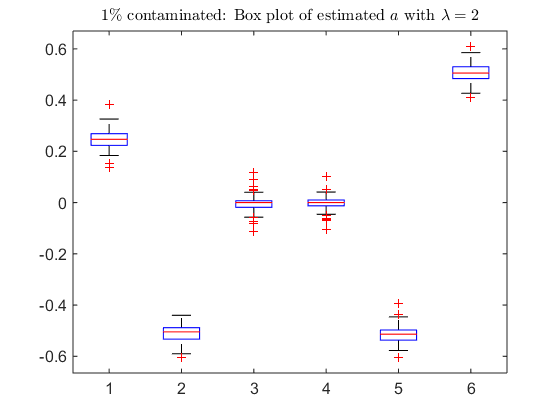}
\caption{Plots of estimated parameters $a_0$ and $a$ for 100 simulations when $1\%$ of the data is contaminated comparing the two scenarios: 1) $\lambda = 50$ which amounts to assuming no outliers in the observed data, and 2) $\lambda=2$ which amounts to assuming there are outliers in the observed data and let the model detect these anomalies. The other parameters used are $r = 1/2$, $\mu = 10$, $s = 1.0$.}
\label{ex2_fig1}
\end{center}
\end{figure}

\begin{figure}
\begin{center}
\includegraphics[scale=0.45]{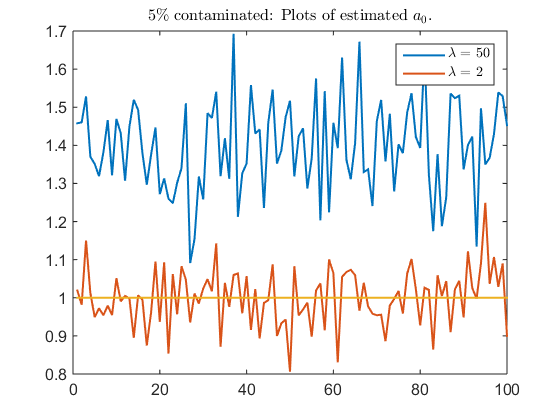}\\
\includegraphics[scale=0.45]{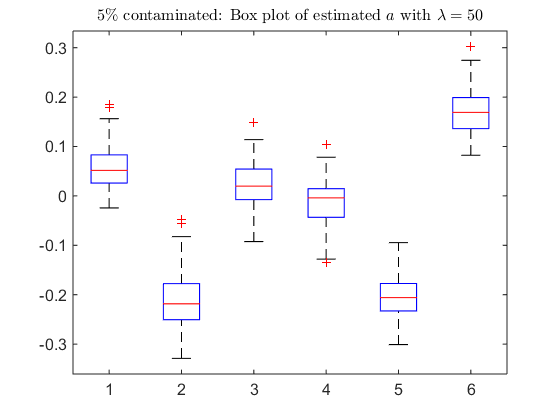}
\includegraphics[scale=0.45]{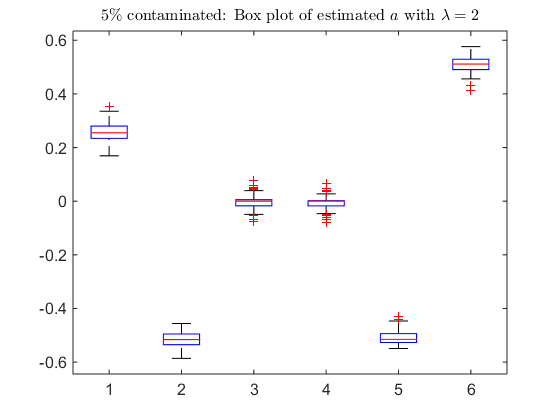}
\caption{Plots of estimated parameters $a_0$ and $a$ for 100 simulations when $5\%$ of the data is contaminated comparing the two scenarios: 1) $\lambda = 50$ which amounts to assuming no outliers in the observed data, and 2) $\lambda=2$ which amounts to assuming there are outliers in the observed data and let the model detect these anomalies. The other parameters used are $r = 1/2$, $\mu = 10$, $s = 1.0$.}
\label{ex2_fig2}
\end{center}
\end{figure}

\begin{figure}
\begin{center}
\includegraphics[scale=0.45]{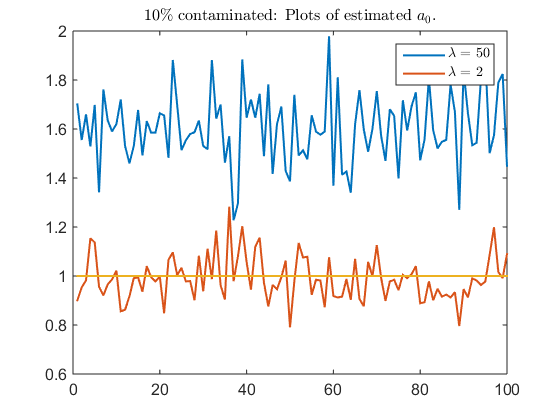}\\
\includegraphics[scale=0.45]{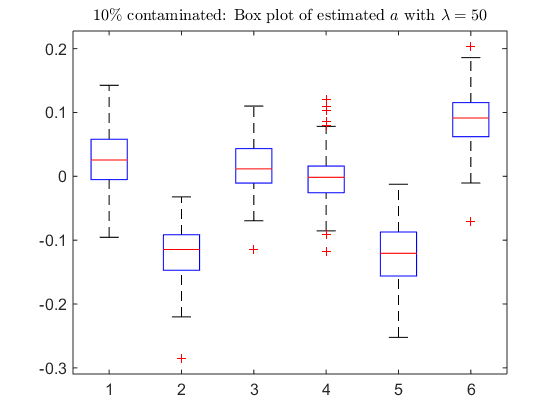}
\includegraphics[scale=0.45]{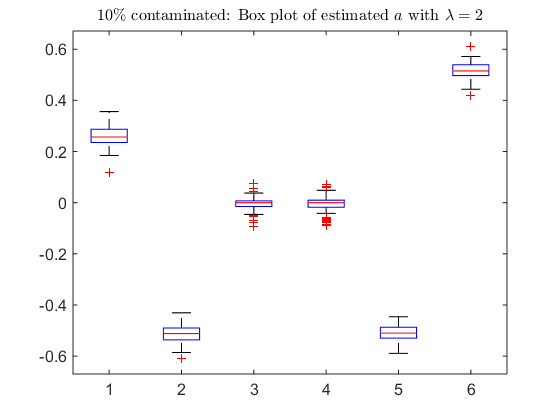}
\caption{Plots of estimated parameters $a_0$ and $a$ for 100 simulations when $10\%$ of the data is contaminated comparing the two scenarios: 1) $\lambda = 50$ which amounts to assuming no outliers in the observed data, and 2) $\lambda=2$ which amounts to assuming there are outliers in the observed data and let the model detect these anomalies. The other parameters used are $r = 1/2$, $\mu = 10$, $s = 1.0$.}
\label{ex2_fig3}
\end{center}
\end{figure}

\begin{figure}
\begin{center}
\includegraphics[scale=0.35]{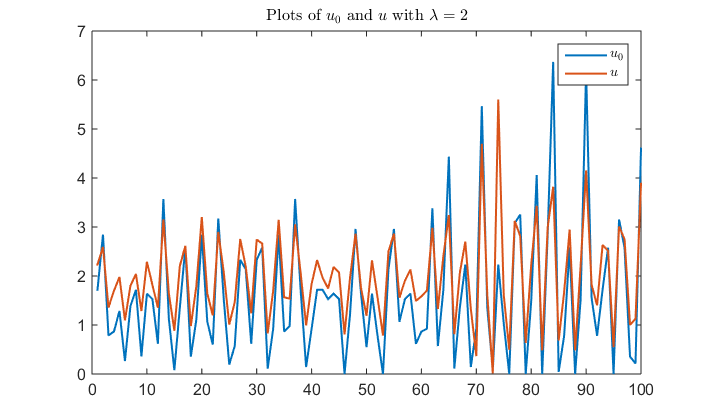}
\includegraphics[scale=0.35]{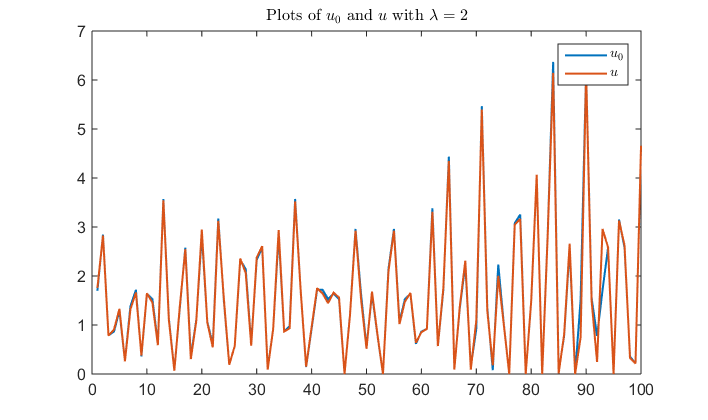}
\caption{Plots of the true mean series $u_0$ and the reconstructed mean series $u$ for one simulation when $1\%$ of the data is contaminated comparing the two scenarios: 1) $\lambda = 50$ which amounts to assuming no outliers in the observed data, and 2) $\lambda=2$ which amounts to assuming there are outliers in the observed data and let the model detect these anomalies. The other parameters used are $r = 1/2$, $\mu = 10$, $s = 1.0$.}
\label{ex2_fig4}
\end{center}
\end{figure}

\begin{figure}
\begin{center}
\includegraphics[scale=0.35]{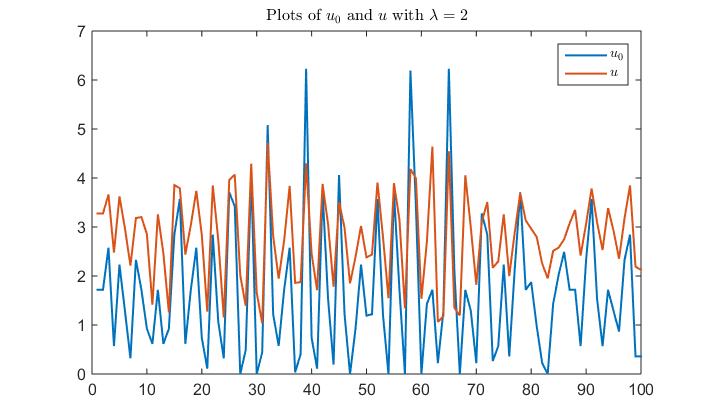}
\includegraphics[scale=0.35]{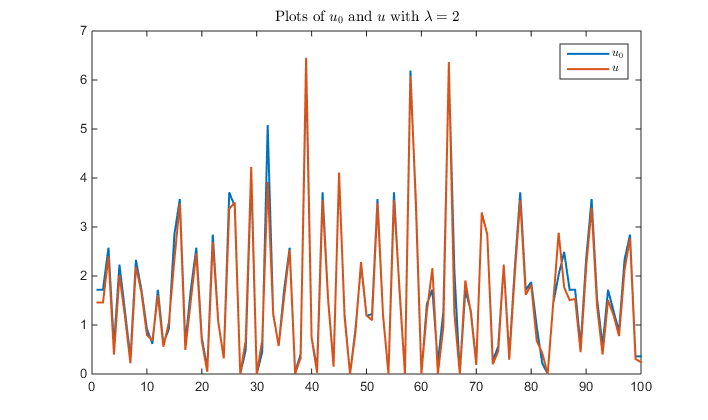}
\caption{Plots of the true mean series $u_0$ and the reconstructed mean series $u$ for one simulation when $5\%$ of the data is contaminated comparing the two scenarios: 1) $\lambda = 50$ which amounts to assuming no outliers in the observed data, and 2) $\lambda=2$ which amounts to assuming there are outliers in the observed data and let the model detect these anomalies. The other parameters used are $r = 1/2$, $\mu = 10$, $s = 1.0$.}
\label{ex2_fig5}
\end{center}
\end{figure}

\begin{figure}
\begin{center}
\includegraphics[scale=0.35]{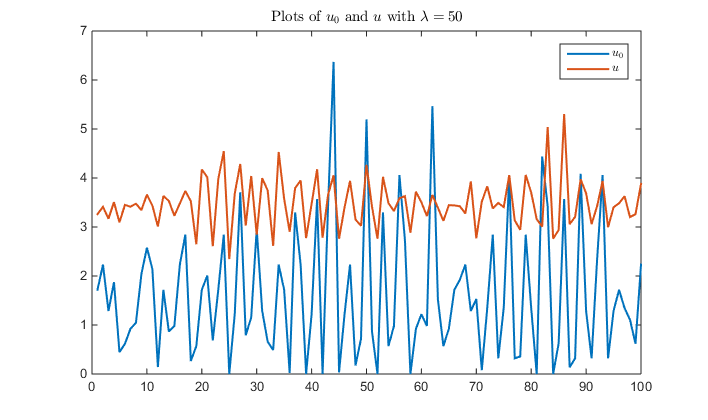}
\includegraphics[scale=0.35]{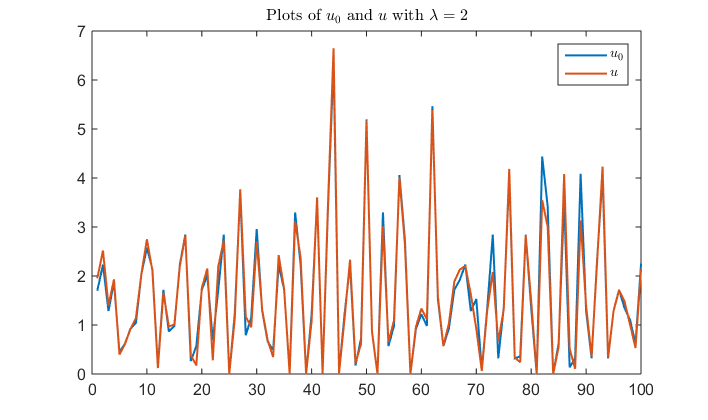}
\caption{Plots of the true mean $u_0$ and the reconstructed mean $u$ for one simulation when $10\%$ of the data is contaminated comparing the two scenarios: 1) $\lambda = 50$ which amounts to assuming no outliers in the observed data, and 2) $\lambda=2$ which amounts to assuming there are outliers in the observed data and let the model detect these anomalies. The other parameters used are $r = 1/2$, $\mu = 10$, $s = 1.0$.}
\label{ex2_fig6}
\end{center}
\end{figure}

\begin{example}
Figures \ref{ex3_fig1}-\ref{ex3_fig2} show numerical results of 100 simulated data that have both missing entries and contamination among the observed ones using Algorithm \ref{alg1}. In Figure \ref{ex3_fig1}, $75\%$ of the entries are observed and $2.5\%$ of these entries are contaminated with outliers. The parameters used are $\lambda = 5, r=0.5, \mu = 30$ and $s = 1$.
 In Figure \ref{ex3_fig2},  $50\%$ of the entries are observed and $2.5\%$ of these entries are contaminated with outliers. The parameters used are $\lambda = 5, r=0.5, \mu = 60$ and $s = 1$.
\end{example}

\begin{figure}
\begin{center}
\includegraphics[scale=0.45]{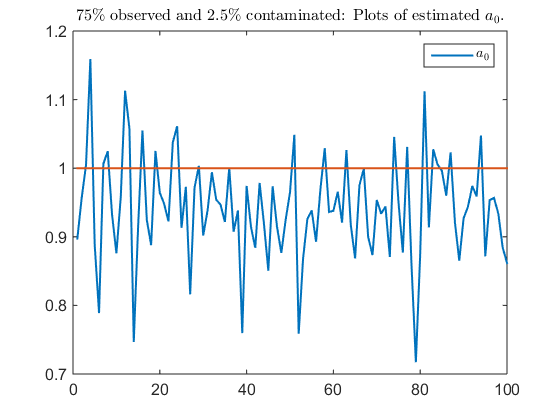}
\includegraphics[scale=0.45]{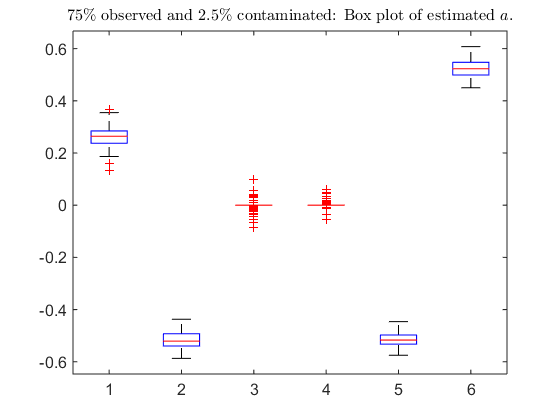}
\caption{Plots of estimated parameters $a_0$ and $a$ for 100 simulations when $75\%$ of the data is observed and $2.5\%$ of the observed entries are contaminated. The parameters used are $\lambda = 5, r=0.5, \mu = 30$ and $s = 1$.}
\label{ex3_fig1}
\end{center}
\end{figure}

\begin{figure}
\begin{center}
\includegraphics[scale=0.45]{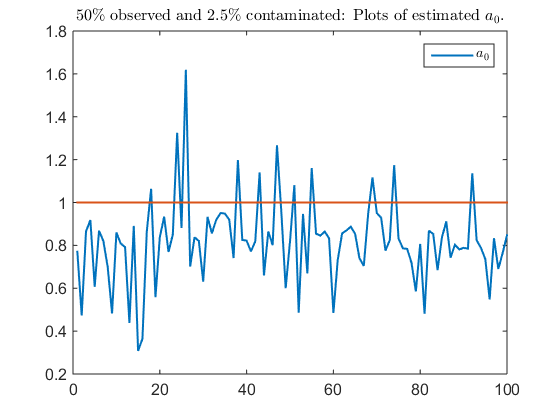}
\includegraphics[scale=0.45]{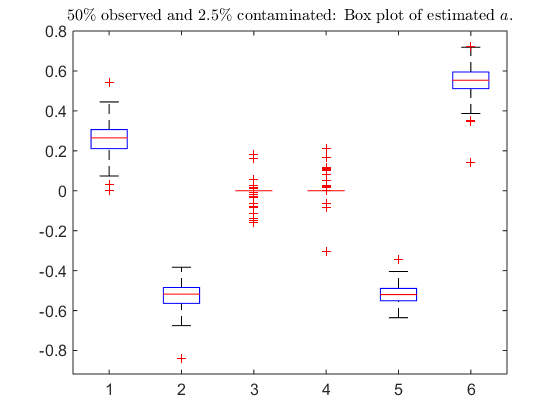}
\caption{Plots of estimated parameters $a_0$ and $a$ for 100 simulations when $50\%$ of the data is observed and $2.5\%$ of the observed entries are contaminated. The parameters used are $\lambda = 5, r=0.5, \mu = 30$ and $s = 1$.}
\label{ex3_fig2}
\end{center}
\end{figure}

\begin{example}
In the proposed model \eqref{pois_min_eq}, $p$ and $q$ are given a prior, and in all previous examples we assume $p=6$ and $q=0$ are known. However in real applications, these values need to be estimated. Model selection, that is picking the right choice for $p$ and $q$, is a difficult problem to tackle. Current approaches involve running the algorithm on various choices of $p$ and $q$ and then use criterion such as AIC \cite{akaike1974new}, BIC \cite{schwarz1978estimating}, etc. to pick the optimal values. We argue that the $\ell^s, 0<s\le 1,$ constraint on the parameters $a$ and $b$ removes the model selection task from the problem and lets the model discover an optimal sparse solution for $a$ and $b$. In this example, we choose $p=15$ and $q=0$. Figure \ref{ex4_fig1} shows the box plots of the estimated parameters $a_0$ and $a$ for 100 simulations. For each simulation, only $75\%$ of entries are observed and $2.5\%$ of the observed entries are contaminated. The parameters used are: $\lambda = 5, r=0.5, \mu = 10,$ and $s = 0.75$.
\end{example}

\begin{figure}
\begin{center}
\includegraphics[scale=0.35]{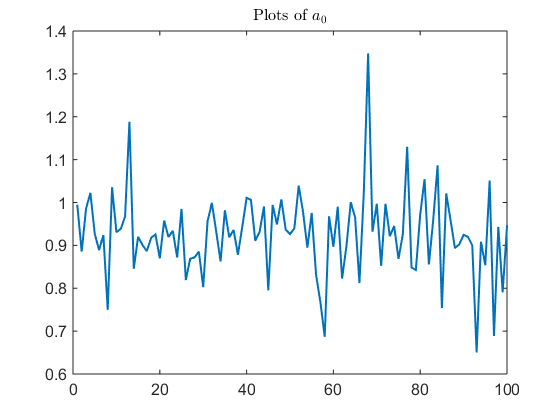}\\
\includegraphics[scale=0.35]{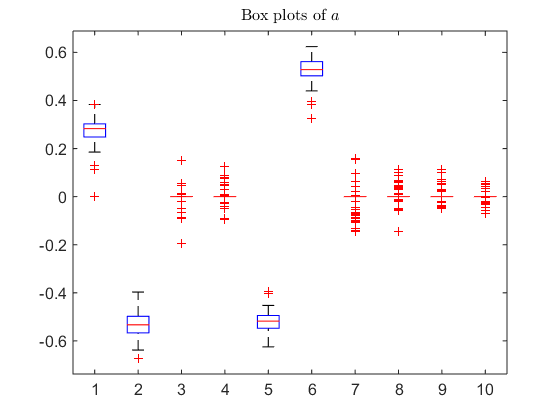}
\includegraphics[scale=0.35]{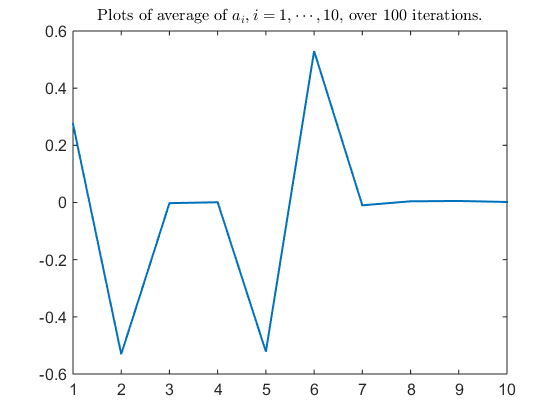}
\caption{Plots of estimated parameters $a_0$ and $a$ for 100 simulations when $75\%$ of the data is observed and $2.5\%$ of the observed entries are contaminated. The parameters used are $\lambda = 5, r=0.5, \mu = 10$ and $s = 0.75$.}
\label{ex4_fig1}
\end{center}
\end{figure}

{\bf Discussion:} In this paper we present an autoregressive time series model to robustly learn the parameters (the mean and correlation coefficients) in the presence of noise, outliers and missing entries.  In the presence of outliers or anomalies, we show that the nonconvex sparsity constraint desensitizes outliers and as a result the model provides a more robust estimation of the parameters. In the presence of missing entries we show that the constraint $\ell^s$, $0<s\le1$, on the parameters significantly improves the accuracy of the estimated parameters. Model selection, that is picking the right choice for $p$ and $q$, is a difficult problem to tackle in time series analysis. Current approaches involve estimating the parameters for various choices of $p$ and $q$ and then use criterion such as AIC, BIC, etc. to pick the optimal $p$ and $q$. The $\ell^s, 0<s\le 1,$ constraint on the parameters $a$ and $b$ removes the model selection task from the problem and lets the model select an optimal sparse solution for $a$ and $b$. However, in return one needs to provide the parameter $\lambda $ and $\mu$ as inputs. We also mention that the proposed model \eqref{min_prob_eq} can be applied to other types of noise besides additive Gaussian or Poisson. Moreover, this model can also be extended to a multivariate case.

\appendix
\section{Appendix}\label{appen_A}

We remark that this is a standard technique for deriving the likelihood function, see for instance \cite{lutkepohl2007new}. For completeness, we show the likelihood function for our problem here in the following proposition.
\begin{proposition}
Given the observed series $\tilde{y}_D=\{\tilde{y}_i\}_{i\in D}$, the minimization problem \eqref{min_prob_eq} is equivalent to
\begin{equation}\label{a_eq0}
\max_{y,f,\theta} P(\tilde{y}_D,y,f,\theta),
\end{equation}
for some fixed $\lambda>0$ and $0<r\le 1$.
\end{proposition}
\begin{proof}
Assume for all $i\in D$, $\tilde{y}_{i}  = y_{i} + x_{i}$, where $x_{i}$ is an additive random noise with
$$
P(x_{i}) = C_r \lambda^{1/r} e^{-\lambda |x_{i}|^r}, \mbox{ for some }\lambda>0\mbox{ and } 0<r\le 1.
$$
This implies
$$
P(\tilde{y}_{i}|y_{i}) = P(x_{i}) = C_r \lambda^{1/r} e^{-\lambda |\tilde{y}_{i}-y_{i}|^r}.
$$
Using Bayes' law, the joint probability is given by
\begin{eqnarray*}
P(\tilde{y}_D,y,f,\theta) &=&P(\{\tilde{y}_{i}\}_{i\in D}, y,f,\theta)\\
&=& P(\tilde{y}_{i}|\{\tilde{y}_{j}\}_{j\in D\setminus\{i\}}, y,f,\theta)P(\{\tilde{y}_{j}\}_{j\in D\setminus\{i\}}, y,f,\theta).
\end{eqnarray*}
Since $\tilde{y}_{i}$ is completely determined by $y_{i}$ and $\lambda$ (given), we have
$$
P(\tilde{y}_{i}|\{\tilde{y}_{j}\}_{j\in D\setminus\{i\}}, y,f,\theta)= P(\tilde{y}_{i}|y_{i}).
$$
This implies
\begin{equation*}
P(\tilde{y}_D,y,f,\theta) = P(\tilde{y}_{i}|y_{i})P(\{\tilde{y}_{j}\}_{j\in D\setminus\{i\}}, y,f,\theta).
\end{equation*}
Recursively apply the same technique to $P(\{\tilde{y}_{j}\}_{j\in D\setminus\{i\}}, y,f,\theta)$, we get
\begin{equation}\label{a_eq1}
P(\tilde{y}_D,y,f,\theta) = \prod_{i\in D} P(\tilde{y}_{i}|y_{i}) P(y,f,\theta).
\end{equation}
As for $P(y,f,\theta)$, we first note that $u_i$ is completely determined by $\{y_j\}_{j<i}$, $f$ and $\theta$. This implies
\begin{eqnarray*}
P(y,f,\theta) &=& P(y_N,\cdots, y_1,f,\theta)\\
&=& P(y_N|y_{N-1},\cdots, y_1,f,\theta) P(y_{N-1},\cdots, y_1,f,\theta).
\end{eqnarray*}
Since $y_N$ is completely determined by $u_N$ and $\theta_N$, we get
$$
P(y_N|y_{N-1},\cdots, y_1,f,\theta) = P(y_N|u_N,\theta_N).
$$
This implies
$$
P(y,f,\theta) = P(y_N|u_N,\theta_N)P(y_{N-1},\cdots, y_1,f,\theta).
$$
Recursively apply the same technique to $P(y_{N-1},\cdots, y_1,f,\theta)$, we get
\begin{equation}\label{a_eq2}
P(y,f,\theta) = \prod_{i=1}^N P(y_i|u_i,\theta_i)P(f,\theta).
\end{equation}
Combining \eqref{a_eq1} and \eqref{a_eq2}, we have
\begin{equation}
P(\tilde{y}_D,y,f,\theta) = \prod_{i\in D} P(\tilde{y}_{i}|y_{i})\prod_{i=1}^N P(y_i|u_i,\theta_i)P(f,\theta),
\end{equation}
where $P(f,\theta)$ is the joint prior on $f$ and $\theta$. 

The maximization problem \eqref{a_eq0} is equivalent to 
\begin{equation}\label{min_llh_eq}
\min_{y,f,\theta} \left\{-\log(P(\tilde{y}_D,y,f,\theta) )\right\},
\end{equation}
where
\begin{equation*}
\begin{split}
-\log(P(\tilde{y}_D,y,f,\theta) ) &= -\sum_{i\in D}\log(P(\tilde{y}_i|y_i)) - \sum_{i=1}^N \log(P(y_i|u_i,\theta_i))\\
& - \log(P(f,\theta)\\
&= -|D|\log(C_r) -\frac{|D|}{r}\log\lambda + \lambda\sum_{i\in D} |\tilde{y}_i-y_i|^r\\
&-\sum_{i=1}^N \log(P(y_i|u_i,\theta)) - \log(P(f,\theta).
\end{split}
\end{equation*}
Treating $r$ and $\lambda$ as fixed constants, we see that \eqref{min_llh_eq} is equivalent to
\begin{equation*}
\begin{split}
\min_{y,f,\theta}\left\{\lambda\sum_{i\in D} |\tilde{y}_i-y_i|^r
-\sum_{i=1}^N \log(P(y_i|u_i,\theta)) - \log(P(f,\theta))\right\},
\end{split}
\end{equation*}
which is \eqref{min_prob_eq} if we assume $f$ and $\theta$ are independent.
\end{proof}

\begin{remark}
Using the same techniques as above, we see that the (-)log-likelihood function corresponding to the energy \ref{pois_min_eq} is given by
\begin{equation}
\begin{split}
L&= \sum_{i=1}^N\left[u_i - y_i\log(u_i) + \log(\Gamma(y_i+1))\right] \\
&-|D|\log(C_r) - \frac{|D|}{r}\log(\lambda) + \lambda \sum_{i\in D} |\widetilde{y}_i - y_i|^r\\
&-(p+q)\log(C_s) - \frac{p+q}{s}\log(\mu) + \mu \left[\sum_{j=1}^p |a_j|^s + \sum_{j=1}^q |b_j|^s\right].
\end{split}
\end{equation}

\end{remark}

\bibliography{ts_biblio}

\begin{thebibliography}{10}

\bibitem{akaike1974new}
Hirotugu Akaike.
\newblock A new look at the statistical model identification.
\newblock {\em Automatic Control, IEEE Transactions on}, 19(6):716--723, 1974.

\bibitem{bassett1978asymptotic}
Gilbert Bassett~Jr and Roger Koenker.
\newblock Asymptotic theory of least absolute error regression.
\newblock {\em Journal of American Statistical Association}, 73(363):618--622,
  1978.

\bibitem{beck2009fast}
Amir Beck and Marc Teboulle.
\newblock A fast iterative shrinkage-thresholding algorithm for linear inverse
  problems.
\newblock {\em SIAM Journal on Imaging Sciences}, 2(1):183--202, 2009.

\bibitem{bolte2014proximal}
Jerome Bolte, Shoham Sabach, and Marc Teboulle.
\newblock Proximal alternating linearized minimization for nonconvex and
  nonsmooth problems.
\newblock {\em Mathematical Programming}, 146(1-2):459--494, 2014.

\bibitem{box1970time}
George~EP Box and Gwilym Jenkins.
\newblock Time series analysis: Forecasting and control.
\newblock {\em Holden-D. iv, San Francisco, 1970}.

\bibitem{candes2009exact}
Ammanuel~J Candes and Benjamin Retch.
\newblock Exact matrix completion via convex optimization.
\newblock {\em Foundations of computational mathematics}, 9(6):717--772, 2009.

\bibitem{candes2006robust}
Emmanuel~J Candes, Justin Romberg, and Terence Tao.
\newblock Robust uncertainty principles: Exact signal reconstruction from
  highly incomplete frequency information.
\newblock {\em Information Theory, IEEE Transactions on}, 52(2):489--509, 2006.

\bibitem{candes2006stable}
Emmanuel~J Candes, Justin~K Romberg, and Terence Tao.
\newblock Stable signal recovery from incomplete and inaccurate measurements.
\newblock {\em Communications on pure and applied mathematics},
  59(8):1207--1223, 2006.

\bibitem{candes2010power}
Emmanuel~J Candes and Terence Tao.
\newblock The power of convex relaxation: Near-optimal matrix completion.
\newblock {\em Information Theory, IEEE Transactions on}, 56(5):2053--2080,
  2010.

\bibitem{chartrand2007exact}
Rick Chartrand.
\newblock Exact reconstruction of sparse signals via nonconvex minimization.
\newblock {\em Signal Processing Letters, IEEE}, 14(10):707--710, 2007.

\bibitem{choi2013investigation}
Yunjin Choi and Roberts Tibshirani.
\newblock An investigation of methods for handling missing data with penalized
  regression.
\newblock {\em arXiv preprint arXiv:1310.2076}, 2013.

\bibitem{dempster1977maximum}
Arthur~P Dempster, Nan~M Laird, and Donald~B Rubin.
\newblock Maximum likelihood from incomplete data via the em algorithm.
\newblock {\em Journal of the Royal statistical society. Series B
  (methodological)}, pages 1--38, 1977.

\bibitem{donoho2006compressed}
David~L Donoho.
\newblock Compressed sensing.
\newblock {\em Information Theory, IEEE Transactions on}, 52(4):1289--1306,
  2006.

\bibitem{finkbeiner1979estimation}
Carl Finkbeiner.
\newblock Estimation for the multiple factor model when data are missing.
\newblock {\em Psychometrika}, 44(4):409--420, 1979.

\bibitem{fokianos2012count}
Konstantinos Fokianos.
\newblock Count time series models.
\newblock {\em Handbook of Statistics}, 30:315--347, 2012.

\bibitem{frank1993statistical}
Ildiko~E Frank and Jerome~H Friedman.
\newblock A statistical view of some chemometrics regression tools.
\newblock {\em Technometrics}, 35(2):109--135, 1993.

\bibitem{ganti2015learning}
Ravi Ganti, Nikhil Rao, Rebecca~M Willett, and Robert Nowak.
\newblock Learning single index models in high dimension.
\newblock {\em arXiv preprint arXiv:1506.08910}, 2015.

\bibitem{george1994time}
Box George.
\newblock {\em Time Series Analysis: Forecasting \& Control, 3/e}.
\newblock Pearson Education India, 1994.

\bibitem{goldstein2014fast}
Tom Goldstein, Brenda O'donoghue, Simon Setzer, and Richard Baraniuk.
\newblock Fast alternating direction optimization methods.
\newblock {\em SIAM Journal on Imaging Sciences}, 7(3):1588--1623, 2014.

\bibitem{harvey1989time}
Andrew~C Harvey and C~Fernandes.
\newblock Time series models for count or qualitative observations.
\newblock {\em Journal of Business \& Economic Statistics}, 7(4):407--417,
  1989.

\bibitem{horton2007much}
Nicholas~J Horton and Ken~P Kleinman.
\newblock Much ado about nothing: A comparison of missing data methods and
  software to fit incomplete data regression models.
\newblock {\em The American Statistician}, 61(1), 2007.

\bibitem{huber1964robust}
Peter~J Huber.
\newblock Robust estimation of a location parameter.
\newblock {\em The Annals of Mathematical Statistics}, 35(1):73--101, 1964.

\bibitem{huber1973robust}
Peter~J Huber.
\newblock Robust regression: asymptotics, conjectures and monte carlo.
\newblock {\em The Annals of Mathematical Statistics}, pages 799--821, 1973.

\bibitem{ilienko2013continuous}
Andrii Ilienko.
\newblock Continuous counterparts of poisson and binomial distributions and
  their properties.
\newblock {\em Annales Univ. Sci. Budapest., Sect. Comp.}, 39:137--147, 2013.

\bibitem{li1994time}
WK~Li.
\newblock Time series models based on generalized linear models: some further
  results.
\newblock {\em Biometrics}, pages 506--511, 1994.

\bibitem{loh2012high}
By~Po-Ling Loh and Martin~J Wainwright.
\newblock High-dimensional regression with noisy and missing data: provable
  guarantees with nonconvexity.
\newblock {\em The Annals of Statistics}, 40(3):1637--1664, 2012.

\bibitem{lutkepohl2007new}
Helmut L{u}tkepohl.
\newblock {\em New Introduction to multiple time series analysis}.
\newblock Springer Science \& Business Media, 2007.

\bibitem{marsaglia1986incomplete}
George Marsaglia.
\newblock The incomplete $\gamma$ function as a continuous poisson
  distribution.
\newblock {\em Computers \& Mathematics with Applications}, 12(5):1187--1190,
  1986.

\bibitem{mateos2012robust}
Gonzalo Mateos and Georgios~B Giannakis.
\newblock Robust nonparametric regression via sparsity control with
  applications to load curve data cleansing.
\newblock {\em Signal Processing, IEEE Transactions on}, 60(4):1572--1584,
  2012.

\bibitem{newman2003longitudinal}
Daniel~A Newman.
\newblock Longitudinal modeling with randomly and systematically missing data:
  A simulation of ad hoc, maximum likelihood, and multiple imputation
  techniques.
\newblock {\em Organizational research methods}, 6(3):328--362, 2003.

\bibitem{nguyen2013robust}
Nam~H Nguyen and Trac~D Tran.
\newblock Robust lasso with missing and grossly corrupted observations.
\newblock {\em IEEE Transactions on Information Theory}, 59(4):2036--2058,
  2013.

\bibitem{nie2012robust}
Feiping Nie, Hua Wang, Xiao Cai, Heng Huang, and Chibiao Ding.
\newblock Robust matrix completion via joint schatten p-norm and lp-norm
  minimization.
\newblock In {\em Data mining (ICDM), 2012 IEEE 12th Internaitonal Conference
  on}, pages 566--574. IEEE, 2012.

\bibitem{peng2006advance}
Chao-Ying~Joanne Peng, Michael Harwell, Shoa-Mann Liou, and Lee~H Ehman.
\newblock Advances in missing data methods and implications for educational
  research.
\newblock {\em Real data analysis}, pages 31--78, 2006.

\bibitem{polson2014bayesian}
Nocholas~G Polson, James~G Scott, and Jesse Windle.
\newblock The bayesian bridge.
\newblock {\em Journal of the Royal Statistical Society: Series B (Statistical
  Methodology)}, 76(4):713--733, 2014.

\bibitem{rubin1976inference}
Donald~B Rubin.
\newblock Inference and missing data.
\newblock {\em Biometrika}, 63(3):581--592, 1976.

\bibitem{schwarz1978estimating}
Gideon Schwarz et~al.
\newblock Estimating the dimension of a model.
\newblock {\em The Annals of Statistics}, 6(2):461--464, 1978.

\bibitem{tibshirani1996regression}
Robert Tibshirani.
\newblock Regression shrinkage and selection via the lasso.
\newblock {\em Journal of the Royal Statistical Society. Series B
  (Methodological)}, pages 267--288, 1996.

\bibitem{trefethen1997numerical}
Lloyd~N Trefethen and David Bau~III.
\newblock {\em Numerical linear algebra}, volume~50.
\newblock SIAM, 1997.

\bibitem{xu2014globally}
Yangyang Xu and Wotao Yin.
\newblock A globally convergent algorithm for nonconvex optimization based on
  block coordinate update.
\newblock {\em arXiv preprint arXiv:1410.1386}, 2014.

\bibitem{zeger1988markov}
Scott~L Zeger and Bahjat Qaqish.
\newblock Markov regression models for time series: A quasi-likelihood
  approach.
\newblock {\em Biometrics}, pages 1019--1031, 1988.

\bibitem{zhongben2012regularization}
Xu~Zhongben, Chang Xiangyu, Xu~Fengmin, and Zhang Hai.
\newblock $l^{1/2}$ regularization: A thresholding representation theory and a
  fast solver.
\newblock {\em IEEE transactions on neural networks and learning systems},
  23(7):1013--1027, 2012.

\end{thebibliography}
\bibliographystyle{plain}

\end{document}